\documentclass[journal,onecolumn]{IEEEtran}
\usepackage{mathrsfs}
\usepackage{url}
\usepackage{psfrag}
\usepackage{graphicx}
\usepackage{fancyhdr,epsfig}
\usepackage{mathtools}
\usepackage{research5IEEE}
\usepackage{amsmath}
\usepackage{upgreek}
\usepackage{amsfonts}
\usepackage{amssymb}
\usepackage{amsthm}
\usepackage{tikz}

\newtheorem{corollary}{Corollary}
\newtheorem{lemma}{Lemma}
\newtheorem{theorem}{Theorem}
\newtheorem{definition}{Definition}
\newtheorem{proposition}{Proposition}
\newtheorem{remark}{Remark}

\newtheorem{note}{Note}

\newcommand{\e}{\textnormal{e}}
\newcommand{\One}{\textnormal{One}}
\newcommand{\MAC}{\textnormal{MAC}}
\newcommand{\BC}{\textnormal{BC}}
\newcommand{\linfb}{\textnormal{linfb}}
\newcommand{\Oz}{\textnormal{Oz}}
\newcommand{\Jafar}{\textnormal{Jafar}}
\newcommand{\fb}{\textnormal{fb}}

\newcommand{\nofb}{\textnormal{nofb}}
\newcommand{\SISO}{\textnormal{SISO}}
\newcommand{\SIMO}{\textnormal{SIMO}}
\newcommand{\MISO}{\textnormal{MISO}}

\newcommand{\B}{\textnormal{B}}
\newcommand{\BHone}{\mat{H}_1^{\B}}
\newcommand{\BHtwo}{\mat{H}_2^{\B}}
\newcommand{\BH}{\mat{H}^{\B}}
\newcommand{\bBHone}{\mat{\bar H}_1^{\B}}
\newcommand{\bBHtwo}{\mat{\bar H}_2^{\B}}
\newcommand{\bBH}{\mat{\bar H}^{\B}}
\newcommand{\tBHone}{\trans{(\mat{ H}_1^{\B})}}
\newcommand{\tBHtwo}{\trans{(\mat{ H}_2^{\B})}}

\newcommand{\tr}[1]{\textnormal{tr}\left(#1\right)}

\title{MIMO MAC-BC Duality  with Linear-Feedback Coding Schemes}
\author{Selma Belhadj Amor, Yossef Steinberg, and Mich\`ele Wigger
\thanks{This work was in part presented at the {\em International Zurich Seminar on Communications}, in Zurich, Switzerland, February 2014.}
\thanks{S.~Belhadj Amor and M.~Wigger are with the Department of Communications and Electronics, Telecom Paristech, Paris, France. email:\{belhadjamor,wigger\}@telecom-paristech.fr.}\thanks{Y.~Steinberg is with the Department of Electrical Engineering at the Technion---Israel Institute of Technology, Haifa, Israel. email:ysteinbe@ee.technion.ac.il.}
\thanks{
The work of S.~Belhadj Amor and M.~Wigger has been supported by the city of Paris under the ``Emergences" program. The work of Y.~Steinberg has been supported by the Israel Science Foundation (grant no.~684/11).}
}

\begin{document}

 \maketitle 

\begin{abstract}
We show that for the multi-antenna Gaussian multi-access channel (MAC) and broadcast channel (BC) with perfect feedback, 
the rate regions achieved by linear-feedback coding schemes (called \emph{linear-feedback capacity regions}) coincide when the same total input-power constraint is imposed on both channels and when the  MAC channel matrices are the transposes of the BC channel matrices. Such a pair of MAC and BC is called \emph{dual}. We also identify sub-classes of linear-feedback coding schemes that achieve the  linear-feedback capacity regions of these two channels and present multi-letter expressions for the linear-feedback capacity regions. Moreover, within the two sub-classes of  coding schemes that achieve the linear-feedback capacity regions for a given MAC and its dual BC, we identify for each MAC scheme   a BC scheme and for each BC scheme  a MAC scheme so that the two schemes have same total input power and  achieve the same rate regions.

In the two-user case, when the  transmitters or the  receiver  are single-antenna, the capacity region for the Gaussian MAC is known \cite{OZAROW84}, \cite{JAFAR06} and the capacity-achieving scheme is a linear-feedback coding scheme. With our results we can thus  determine the linear-feedback capacity region of the two-user Gaussian BC when either  transmitter or  receivers are single-antenna and we can identify the corresponding linear-feedback capacity-achieving coding schemes. Our results show that the control-theory inspired linear-feedback coding scheme by  Elia~\cite{Elia04},  by Wu~{\it et al.}~\cite{Wu05}, and by Ardestanizadeh~{\it et al.} \cite{AMF12} is sum-rate optimal among all linear-feedback coding schemes for the symmetric single-antenna Gaussian BC with equal channel gains. 

In the $K\geq 3$-user case, Kramer \cite{Kramer02} and Ardestanizadeh~{\it et al.}~\cite{AWKJ} determined the linear-feedback sum-capacity for the symmetric single-antenna Gaussian MAC with equal channel gains. Using our duality result, in this paper we  identify the linear-feedback sum-capacity for the $K\geq 3$-user single-antenna Gaussian BC with equal channel gains. It is equal to the sum-rate achieved by Ardestanizadeh~{\it et al.}'s linear-feedback coding scheme \cite{AMF12}. 

Our results extend also to the setup where only a subset of the  feedback links are present.

\end{abstract}

\section{Introduction}
\label{sec:intro}
Unlike for point-to-point channels, in multi-user networks feedback can enlarge  capacity. 
For most multi-user networks the capacity region with feedback is however still unknown. Notable exceptions are the two-user memoryless single-input single-output (SISO) Gaussian  multi-access channel (MAC) whose capacity region with feedback was determined by Ozarow~\cite{OZAROW84}, and the two-user single-input multi-output (SIMO) and  multi-input single-output (MISO) memoryless Gaussian MACs, whose capacity regions were determined by Jafar and Goldsmith \cite{JAFAR06}. For more than two users or in the general multi-input multi-output (MIMO) case, the capacity region of the memoryless Gaussian MAC with feedback is still open. For $K>2$ transmitters, Kramer \cite{Kramer02} determined the sum-capacity of the SISO Gaussian MAC under equal power constraints $P$ at all the transmitters when this  $P$ is sufficiently large. 

 Ozarow's coding scheme \cite{OZAROW84}, which achieves the capacity region of the two-user SISO Gaussian MAC with feedback, is a variation of the Schalkwijk-Kailath scheme for point-to-point channels. Each transmitter maps its message to a message point and sends this message point during one of the first two channel uses. In  channel uses $3$ and thereafter both transmitters send scaled versions of the linear minimum mean squared estimation (LMMSE) errors of their message points when observing all previous outputs. Ozarow showed that this scheme achieves the sum-capacity of the two-user SISO Gaussian MAC with perfect feedback. To achieve the entire capacity region, one of the two transmitters has to combine this scheme with a nofeedback scheme using rate-splitting. 
The described scheme falls into the class of \emph{linear-feedback coding schemes} \cite{AWKJ}, 
where the transmitters can use the feedback signals only in a linear way. That means, a transmitter's channel input for a given channel use is a \emph{linear} combination of the previously observed feedback signals and some information-carrying code symbols
which only depend on the transmitter's message but not on the feedback. 

Jafar and Goldsmith's \cite{JAFAR06} capacity-achieving schemes for the two-user SIMO and MISO Gaussian MACs and Kramer's scheme for the $K$-user SISO Gaussian MAC are variations of Ozarow's scheme and also belong to the class of linear-feedback coding schemes. It has recently been shown \cite{AWKJ} that 
under equal input-power constraints $P$ at all $K$ transmitters, irrespective of the values of $P$ and $K$,  Kramer's scheme achieves the largest sum-rate among all  linear-feedback coding schemes.

The capacity region of the memoryless Gaussian BC with perfect feedback is unknown even with only two receivers and in the SISO case. Achievable regions  have been proposed by Ozarow \& Leung~\cite{OZAROW-LEUNG}, Elia~\cite{Elia04},  Kramer~\cite{Kramer02}, Wu~{\it et al.}~\cite{Wu05}, Ardestanizadeh~{\it et al.}~\cite{AMF12},  Gastpar~{\it et al.}~\cite{GLSW11}, Wu \& Wigger~\cite{WU}, Shayevitz \& Wigger~\cite{SHAYEVITZ}, and Venkataramanan \& Pradhan~\cite{VENKATARAMAN}. The schemes in \cite{OZAROW-LEUNG}, \cite{Elia04},  \cite{Kramer02}, \cite{Wu05}, \cite{AMF12},  \cite{GLSW11} are linear-feedback coding schemes and  outperform the other schemes \cite{WU}, \cite{SHAYEVITZ}, \cite{VENKATARAMAN} when these latter are specialized to the SISO Gaussian BC.  
For example, for some setups where the noises at the two receivers are correlated, the scheme in~\cite{GLSW11} provides the largest achievable rates known to date. 
Also, in the asymptotic regime where the allowed input power $P\to \infty$ it achieves the sum-capacity, irrespective of the correlation between the noise sequences at the two receivers. 

For finite input power $P$ and when the two noise sequences are uncorrelated, the largest achievable sum-rate known  to date for the \emph{symmetric} SISO Gaussian BC is attained by 
the linear-feedback coding schemes in \cite{Elia04}, \cite{Wu05}, and \cite{AMF12}, which are designed based on control-theoretic considerations.  Specifically, they achieve the same sum-rate over the symmetric SISO Gaussian BC  under power constraint $P$ as Ozarow's scheme \cite{OZAROW84} achieves over the Gaussian  MAC under a \emph{sum-power} constraint $P$. Thus, there is a duality in terms of achievable sum-rate between the control-theoretic schemes for the BC in \cite{Elia04}, \cite{Wu05}, and \cite{AMF12} and Ozarow's capacity-achieving scheme for the MAC. It is unknown whether the  schemes in \cite{Elia04}, \cite{Wu05}, and \cite{AMF12} achieve the sum-capacity with perfect feedback for symmetric BCs, and previous to this work, it was also unknown  whether for the symmetric SISO Gaussian BC it is  sum-rate optimal among all linear-feedback coding schemes. As detailed shortly, our results in this paper show that this is indeed the case.

Without feedback, the following {duality relation} is well known~\cite{VJG03,VisTse,WSS06}: under the same sum input-power constraint the capacity regions of the MIMO Gaussian MAC and BC coincide when the channel matrices of the MAC and BC are transposes of each other. 
Such a pair of MAC and BC is called \emph{dual}. 

Our main contribution in this work is the following new duality result:  
with perfect feedback and when restricting to  linear-feedback coding schemes, the set of all achievable rates, coincide for 
the MIMO Gaussian MAC and BC when the two channels are dual and when the same sum input-power constraint is imposed on their inputs.
This result is particularly interesting in the two-user case and when either transmitter(s) or receiver(s) are single-antenna (SISO, MISO, and SIMO setups) because for these setups computable single-letter characterizations
of the linear-feedback capacity regions of the Gaussian MAC are known. With our duality result, we thus  immediately obtain single-letter characterizations of the linear-feedback capacity regions for the two-user SISO, SIMO, and MISO Gaussian BC. For more than $K\geq 3$ users the linear-feedback sum-capacity of the SISO Gaussian MAC is known when the channel gains are equal \cite{Kramer02, AWKJ}; with our results we thus obtain the linear-feedback sum-capacity of the SISO Gaussian BC when the channel gains are equal. Our results in particular show that the control-theory inspired linear-feedback coding schemes  proposed by  Elia~\cite{Elia04},  by Wu~{\it et al.}~\cite{Wu05}, and by Ardestanizadeh~{\it et al.} \cite{AMF12} are sum-rate optimal among all linear-feedback coding schemes for the symmetric SISO Gaussian BC with equal channel gains, irrespective of the number of receivers $K\geq 2$. 

We also introduce a class of (multi-letter) linear-feedback schemes for the MIMO Gaussian MAC and BC that achieve the linear-feedback capacity regions. Within this class we can identify the pairs of schemes  that achieve the same rate-regions over dual MACs and BCs. 
Since we know the optimal linear-feedback schemes for the two-user SISO, SIMO, and MISO Gaussian MAC \cite{OZAROW84,JAFAR06}, we can identify the optimal linear-feedback schemes for the two-user SISO, MISO, and SIMO Gaussian BC. 

Our results extend also to a setup where only some of the feedback links are present. 

The remainder of this paper is organized as follows. 
In Section~\ref{sec:prels} we explain the notations used in this paper and  introduce some preliminaries. In 
Section~\ref{sec:BC}, we consider the two-user MIMO Gaussian BC with perfect feedback and in Section~\ref{sec:MAC}  the two-user MIMO Gaussian MAC with perfect feedback: specifically, we describe the channel model, introduce the class of linear-feedback coding schemes, and summarize previous results.
Section~\ref{sec:main} presents our main results on MAC-BC duality with linear-feedback schemes, 
and  the linear-feedback capacity-achieving schemes for MAC and BC. In Sections~\ref{sec:ext_one} and~\ref{sec:ext_many}, we explain how our results extend to setups with partial feedback and to arbitrary $K\geq 2$  users.
Finally, Section~\ref{sec:proofs} contains the major proofs.

\section{Notation and Preliminaries}
\label{sec:prels}
In the following, a random variable is denoted by an upper-case letter (e.g $X$, $Y$, $Z$) and its realization by a lower-case letter (e.g $x$, $y$, $z$).
 An $n$-dimensional random column-vector and its realization are denoted by boldface symbols (e.g.~$\vect{X}$, $\vect{x}$). We use $\|\cdot\|$ to indicate the Euclidean norm and $\E{\cdot}$ for the expectation operator.  The abbreviation i.i.d. stands for \emph{independently and identically distributed}.
 
Sets are denoted by calligraphic letters (e.g., $\set{X}$, $\set{Y}$, $\set{Z}$) and $\set{X}\times \set{Y}$ denotes the Cartesian product of the sets  $\set{X}$ and $\set{Y}$.  The set of real numbers is denoted by $\Reals$ and its $d$-fold Cartesian product by $\Reals^d$. We use $\textnormal{cl}(X)$ to denote the convex closure of the set $\set{X}$.

Throughout the paper, $\log(\cdot)$ refers to the binary logarithm-function.
 
 To denote matrices we use the font $\mat{A}$. For the transpose of a matrix $\mat{A}$ we write $\trans{\mat{A}}$, for its determinant $|\mat{A}|$, and for its trace $\textnormal{tr}(\mat{A})$. For the Kronecker product of two matrices $\mat{A}$ and $\mat{B}$ we write $\mat{A}\otimes \mat{B}$. We use
 $\mat{I}_d$ to denote the $d$-by-$d$ identity matrix, where we  drop the subscript whenever the dimensions are clear from the context. The symbol $\mat{E}_d$ denotes the $d$-by-$d$ \emph{exchange matrix} which is 0 everywhere except on the counter-diagonal where it is 1. For example, 
 \begin{equation}
 \mat{E}_3= \begin{bmatrix} 0 &0 & 1\\
  0& 1& 0\\ 
  1& 0& 0\end{bmatrix}.
 \end{equation}
 Again, we drop the subscript whenever the dimensions are clear.
 We can now define the \emph{reverse image matrix operator} $\mat{\bar{\cdot}}$: For a given  $d_1$-by-$d_2$ matrix $\mat{A}$, 
 \begin{equation}\label{eq:reverse}
 \mat{\bar{A}}\eqdef \mat{E}_{d_2} \trans{\mat{A}}\mat{E}_{d_1}. 
 \end{equation}

\begin{note}\label{note:bar}
The reverse image matrix operator satisfies the following properties: 
\begin{enumerate}
\item Applying the operator twice results in the identity operation: $\mat{A}= \mat{\bar{\mat{\bar{A}}}}$.
\item The operator commutes with the matrix inverse-operator and the product operator: 
\begin{IEEEeqnarray}{rCl}
(\mat{\bar{A}})^{-1} &=& \overline{\mat{(\mat{A}^{-1})}}\\
\mat{\bar{A}}\mat{\bar{B}}&=&\mat{ \overline{(\mat{B}\mat{A})}}.
\end{IEEEeqnarray}
\item The operator maps a strictly-lower block-triangular $\eta\kappa_1$-by-$\eta\kappa_2$ matrix of block sizes $\kappa_1\times \kappa_2$ into a strictly-lower block-triangular $\eta \kappa_2$-by-$\eta\kappa_1$ matrix of block sizes $\kappa_2\times \kappa_1$. 
\end{enumerate}
\end{note}

\section{MIMO Gaussian BC with Feedback}
\label{sec:BC}
\subsection{Setup}
\label{sec:setupBC}
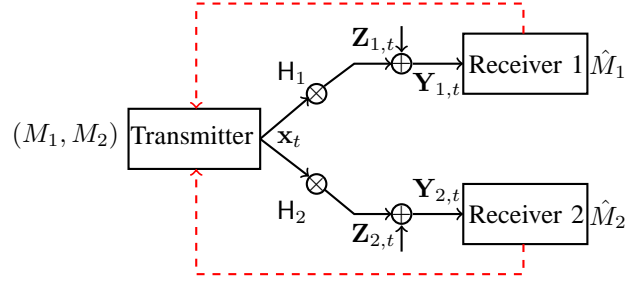
\begin{figure}[ht]
\centerline{
\begin{tikzpicture}[xscale=1,yscale=1]
\draw [thick] (0,2.1) rectangle (1.75,2.9);
\node[right] at (-0.1,2.55){Transmitter};
\node[left] at (0,2.55){$(M_1,M_2)$};
\draw [->][thick](1.75,2.5)--(2.4,3.05);
\draw [thick](2.5,3.1) circle [radius=0.13];
\node[above left] at (2.5,3.1){$\mat{H}_1$};
\node  at (2.5,3.1){$\times$};
\draw [->][thick](2.6,3.2)--(3,3.5)--(3.5,3.5);
\node[right] at (1.85,2.5){$\vect{x}_t$};
\draw [thick](3.63,3.5) circle [radius=0.13];
\node at (3.63,3.5){$+$};
\draw [->][thick](3.63,4)--(3.63,3.65);
\node [left] at (3.68,3.8){$\vect{Z}_{1,t}$};
\draw [->][thick](3.78,3.5)--(4.45,3.5);
\draw [thick] (4.45,3.1) rectangle (6.1,3.9);
\node[right] at (4.4,3.5){Receiver~1};
\node[right] at (6,3.5){$\hat{M}_1$};
\draw [->][thick](1.75,2.5)--(2.4,2);
\draw [thick](2.5,1.9) circle [radius=0.13];
\node[below left] at (2.5,1.8){$\mat{H}_2$};
\node  at (2.5,1.9){$\times$};
\draw [->][thick](2.6,1.83)--(3,1.5)--(3.5,1.5);
\draw [thick](3.63,1.5) circle [radius=0.13];
\node at (3.63,1.5){$+$};
\draw [->][thick](3.63,1)--(3.63,1.35);
\node [left] at (3.68,1.2){$\vect{Z}_{2,t}$};
\draw [->][thick](3.78,1.5)--(4.45,1.5);
\draw [thick] (4.45,1.1) rectangle (6.1,1.9);
\node[right] at (4.4,1.5){Receiver~2};
\node[right] at (6,1.5){$\hat{M}_2$};
\node[below] at (4.1,3.5){$\vect{Y}_{1,t}$};
\node[above] at (4.1,1.5){$\vect{Y}_{2,t}$};
\draw [->][dashed,thick][red](5.25,3.9)--(5.25,4.3)--(0.9,4.3)--(0.9,2.9);
\draw [->][dashed,thick][red](5.25,1.1)--(5.25,0.7)--(0.9,0.7)--(0.9,2.1); 
\end{tikzpicture}}
\caption{Two-user MIMO Gaussian BC with feedback.\label{fig:bc}} 
\end{figure}
We consider the two-user memoryless MIMO Gaussian BC with perfect-output feedback depicted in Figure~\ref{fig:bc}. The transmitter is equipped with $\kappa$ transmit-antennas and each Receiver $i$, 
for $i\in \{1,2\}$, is equipped with $\nu_i$ receive-antennas.
At each time $t\in \mathbb{N}$, if $\vect{x}_t$ denotes the  real vector-valued input symbol sent by the transmitter, 
Receiver~$i\in\{1,2\}$ observes the real vector-valued channel output
\begin{IEEEeqnarray}{rCL}\label{eq:BCmodel}
\vect{Y}_{i,t}&=&\mat{H}_i \vect{x}_t+\vect{Z}_{i,t},
\end{IEEEeqnarray}
where $\mat{H}_i$, for $i\in\{1,2\}$, is a deterministic real $\nu_i$-by-$\kappa$ channel matrix known to transmitter and receivers and $\{\vect{Z}_{1,t}\}_{t=1}^n$ and $\{\vect{Z}_{2,t}\}_{t=1}^n$ are independent sequences of i.i.d.~centered Gaussian random vectors of identity covariance matrix.

The transmitter wishes to convey a message $M_1$ to Receiver~1 and an independent message $M_2$ to Receiver~2. The messages are independent of the noise sequences $\{\vect{Z}_{1,t}\}_{t=1}^n$ and $\{\vect{Z}_{2,t}\}_{t=1}^n$ and uniformly distributed over the sets $\mathcal{M}_1\triangleq\{1,\dots, \lfloor 2^{n R_1}\rfloor\}$ and $\mathcal{M}_2\triangleq\{1,\dots, \lfloor 2^{n R_2}\rfloor\}$, where 
$R_1$ and $R_2$ denote the rates of transmission and $n$ the blocklength. 

The transmitter observes causal, noise-free output feedback from both receivers. Thus, the time-$t$ channel input $\vect{X}_t$ can depend on both messages $M_1$ and $M_2$ and on all previous channel outputs $ \vect{Y}_{1,1},\ldots, \vect{Y}_{1,t-1}$ and $\vect{Y}_{2,1}, \ldots, \vect{Y}_{2,t-1}$:
\begin{IEEEeqnarray}{rCL}\label{eq:inputsBC}
\vect{X}_t=\varphi_t^{(n)}(M_1,M_2, \vect{Y}_{1,1},\ldots, \vect{Y}_{1,t-1},\vect{Y}_{2,1}, \ldots, \vect{Y}_{2,t-1}), \quad t\in\{1,\dots,n\},\IEEEeqnarraynumspace
\end{IEEEeqnarray}
for some encoding function of the form:
\begin{IEEEeqnarray}{rCL}
\varphi_t^{(n)}&:&\mathcal{M}_1 \times \mathcal{M}_2 \times \mathbb{R}^{\nu_1 (t-1)}\times\mathbb{R}^{\nu_2(t-1)}\to \mathbb{R}^\kappa.\label{eq:bc_encoding}
\end{IEEEeqnarray}

We impose an \emph{expected average block-power constraint}
\begin{IEEEeqnarray}{C}
\label{eq:powerconstraintBC}
\frac1n \sum_{t=1}^n \mathbf{E}[ \|\vect{X}_t\|^2]\leq P,
\end{IEEEeqnarray}
where the expectation is over the messages and the realizations of the channel.

Each Receiver~$i$ decodes its corresponding message $M_i$ by means of a decoding function $\phi_{i}^{(n)}$ of the form 
\begin{IEEEeqnarray}{rCl}
 \phi_{i}^{(n)}&\colon& \mathbb{R}^{\nu_i n}\to \mathcal{M}_i,\quad i\in \{1,2\}\label{eq:bc_decoding}.
 \end{IEEEeqnarray}
That means,  based on the output sequence $\vect{Y}_{i,1},\dots,\vect{Y}_{i,n}$, Receiver~$i$ produces the guess 
\begin{equation}
\hat{M}_i^{(n)}=\phi_{i}^{(n)}(\vect{Y}_{i,1},\dots,\vect{Y}_{i,n}).
\end{equation}

An error occurs in the communication if 
\begin{equation}
 (\hat{M}_1\neq M_1) \text{ or } (\hat{M}_2 \neq M_2).
\end{equation}

Thus, the average probability of error is 
\begin{IEEEeqnarray}{rCL} 
 P_{\e,\BC}^{(n)}&\triangleq&\textnormal{Pr}\big[(\hat{M}_1\neq M_1) \text{ or } (\hat{M}_2 \neq M_2)\big].
\end{IEEEeqnarray}

A \emph{$(\lfloor2^{nR_1}\rfloor,\lfloor 2^{nR_2}\rfloor,n)$ MIMO BC feedback-code of 
power $P$} is composed of a sequence of encoding functions $\{g_t^{(n)}\}_{t=1}^n$ as in~\eqref{eq:bc_encoding} and satisfying~\eqref{eq:powerconstraintBC} and of two decoding functions $\phi_1^{(n)}$ and $\phi_2^{(n)}$ as in~\eqref{eq:bc_decoding}.

We say that a rate-pair $(R_1,R_2)$ is achievable over the MIMO Gaussian BC with feedback under a power constraint $P$, if there exists a sequence of $\{(\lfloor 2^{n R_1}\rfloor ,\lfloor 2^{n R_2}\rfloor,n)\}_{n=1}^{\infty}$ MIMO BC feedback-codes such that the average probability of error $P_{\e,\BC}^{(n)}$ tends to zero as the blocklength tends to infinity. 
The closure of the union of all achievable regions is called \emph{capacity region}. 
We denote it by $\set{C}_{\BC}^{\fb}(\mat{H}_1,\mat{H}_2,P)$. The supremum of the sum $R_1+R_2$, where $(R_1,R_2)$ are in $\set{C}_{\BC}^{\fb}(\mat{H}_1,\mat{H}_2,P)$ is called \emph{sum-capacity} and is denoted $C_{\BC,\Sigma}^{\fb}(\mat{H}_1,\mat{H}_2,P)$.
\subsection{Linear-feedback schemes for MIMO BC}
We restrict attention to \emph{linear-feedback  coding schemes} where the transmitter's channel input is a \emph{linear}  combination of the previous feedback signals and an information-carrying vector that depends only on the messages $(M_1,M_2)$ (but not on the feedback). 
Specifically, we assume that the channel input vector has the form
\begin{IEEEeqnarray}{rCl}
\vect{X}_t&=&\vect{W}_t+\sum_{i=1}^2 \sum_{\tau=1}^{t-1} \mat{A}_{i,\tau,t}\vect{Y}_{i,\tau}, \quad t\in \{1,\dots,n\}, \label{eq:lin_com_1_bc}
\end{IEEEeqnarray}
where
 $\vect{W}_t=\xi_t^{(n)}(M_1,M_2)$
and where $\{\mat{A}_{i,\tau,t}\}$ are arbitrary $\kappa$-by-$\nu_i$ matrices.

The mappings $\big\{\xi_t^{(n)}\colon\set{M}_1 \times \set{M}_2\to \Reals^\kappa\big\}_{t=1}^n$ and the decoding operations $\phi_{1}^{(n)}$ and $\phi_{2}^{(n)}$ can be arbitrary.

Taking a linear combination of the information-carrying vector $\vect{W}_t$ and the past output vectors $ \vect{Y}_{1,1},\ldots, \vect{Y}_{1,t-1}$ and $\vect{Y}_{2,1},\ldots,$ $\vect{Y}_{2,t-1}$ is equivalent to taking a (different) linear combination of (a different information-carrying vector) $\vect{\tilde W}_t$ and the past noise vectors  $ \vect{Z}_{1,1},\ldots, \vect{Z}_{1,t-1}$ and $\vect{Z}_{2,1}, \ldots, \vect{Z}_{2,t-1}$. 
Hence, we can equivalently write \eqref{eq:lin_com_1_bc} as
\begin{IEEEeqnarray}{rCl}
\vect{X}_t&=&\vect{\tilde  W}_t+\sum_{i=1}^2 \sum_{\tau=1}^{t-1} \mat{B}_{i,\tau,t}\vect{Z}_{i,\tau}, \quad t\in \{1,\dots,n\}, \label{eq:lin_com_2_bc}
\end{IEEEeqnarray}
where $\vect{\tilde  W}_t=\tilde {\xi}_t^{(n)}(M_1,M_2)$,
for some arbitrary function $\tilde{\xi}_t^{(n)}\colon \set{M}_1\times\set{M}_2\to \Reals^\kappa$, and $\{\mat{B}_{i,\tau,t}\}$ are arbitrary $\kappa$-by-$\nu_i$ matrices.

The set of all rate-pairs achieved by linear-feedback schemes is called  \emph{linear-feedback capacity region} and is denoted $\set{C}_{\BC}^{\linfb}({\mat{H}_1}, {\mat{H}_2};P)$. 
The largest sum-rate achieved by a linear-feedback scheme is called  \emph{linear-feedback sum-capacity} and is denoted $C_{\BC,\Sigma}^{\linfb}({\mat{H}_1},{\mat{H}_2};P)$.

\subsection{Previous Results}
Without feedback, the capacity region of the MIMO Gaussian BC,   $\set{C}_\BC^\nofb\left(\mat{H}_1,\mat{H}_2;P\right)$ was determined by Weingarten, Steinberg, and Shamai~\cite{WSS06}.\footnote{Recently, Nair has extended their result to also allow for an additional common message to be sent to the two receivers.}

With feedback, the capacity region is unknown even in the scalar case. 
Achievable regions---based on linear-feedback schemes---have been proposed in~\cite{Kramer02,OZAROW-LEUNG,Elia04,Wu05,AMF12,GLSW11}. 
Non-linear feedback schemes have been proposed in \cite{WU,SHAYEVITZ,VENKATARAMAN}. 
The best known achievable regions are due to {linear-feedback schemes}. 


\section{MIMO Gaussian MAC with Feedback}\label{sec:MAC}
\subsection{Setup}
\begin{figure}[ht]
\centerline{
\begin{tikzpicture}[xscale=1,yscale=1]
\draw [thick] (-0.5,3.1) rectangle (1.5,3.9);
\node[right] at (-0.55,3.5){Transmitter~1};
\node[left] at (-0.5,3.5){$M_1$};
\draw [->][thick](1.5,3.5)--(2.35,3.5);
\draw [->][thick](2.6,3.5)--(3.05,2.6);
\draw [thick](2.45,3.5) circle [radius=0.13];
\node[above right] at (2.5,3.5){$\trans{\mat{H}_1}$};
\node [right] at (2.2,3.5){$\times$};
\node[above] at (2,3.5){$\vect{x}_{1,t}$};
\draw [->][thick](1.5,1.5)--(2.35,1.5);
\draw [->][thick](2.6,1.5)--(3.05,2.44);
\draw [thick](2.45,1.5) circle [radius=0.13];
\node [right] at (2.2,1.5){$\times$};
\node[below right] at (2.5,1.5){$\trans{\mat{H}_2}$};
\node[below] at (2,1.5){$\vect{x}_{2,t}$};
\draw [thick] (-0.5,1.1) rectangle (1.5,1.9);
\node[right] at (-0.55,1.5){Transmitter~2};
\node[left] at (-0.5,1.5){$M_2$};
\draw [thick] (4,2.1) rectangle (5.5,2.9);
\node[right] at (4,2.55){Receiver};
\node[right] at (5.4,2.55){$(\hat{M}_1,\hat{M}_2)$};
\draw [->][thick](3.3,2.5)--(4,2.5);
\node[below] at (3.7,2.5){$\vect{Y}_t$};
\draw [thick](3.15,2.5) circle [radius=0.13];
\node  at (3.15,2.5){$+$};
\draw [->][thick](3.15,3.18)--(3.15,2.62);
\node [right] at (3.15,3.02){$\vect{Z}_t$};
\draw [->][dashed,thick][red](4.75,2.9)--(4.75,4.3)--(0.55,4.3)--(0.55,3.9);
\draw [->][dashed,thick][red](4.75,2.1)--(4.75,0.7)--(0.55,0.7)--(0.55,1.1);
\end{tikzpicture}}
\caption{Two-user MIMO Gaussian MAC with feedback.\label{fig:mac}}
\end{figure}
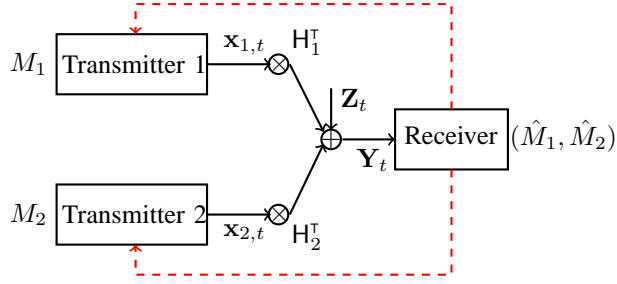

We consider the two-user memoryless MIMO Gaussian MAC with perfect output-feedback in Figure~\ref{fig:mac}. 
Each Transmitter $i$, for $i\in\{1,2\}$, is equipped with $\nu_i$ transmit-antennas and the receiver is equipped with $\kappa$ receive-antennas.
At each time $t\in \mathbb{N}$, if $\vect{x}_{1,t}$ and $\vect{x}_{2,t}$ denote the vector signals sent by Transmitters 1 and~2, the receiver observes the real vector-valued channel output 
\begin{equation}
\vect{Y}_t=\trans{\mat{H}_1} \vect{x}_{1,t}+\trans{\mat{H}_2} \vect{x}_{2,t}+\vect{Z}_t,\label{eq:mac_channel_output}
\end{equation}
where $\mat{H}_i$, for $i\in\{1,2\}$, is a deterministic real $\nu_i$-by-$\kappa$ channel matrix known to transmitters and receiver and $\{\vect{Z}_t\}$ is a sequence of independent and identically distributed $\kappa$-dimensional centered 
Gaussian random vectors of identity covariance matrix. 

The goal of communication is that Transmitters~1~and~2 convey the independent messages $M_1$ and $M_2$ to the common receiver, {where the pair $(M_1,M_2)$ is  independent of the noise sequence $\{\vect{Z}_t\}$.}
{(Recall that $M_i$ is uniformly distributed over $\mathcal{M}_i=\{1,\dots, \lfloor 2^{n R_i}\rfloor\}$)}

The two transmitters observe perfect feedback from the channel outputs. Thus, Transmitter~$i$'s, $i\in\{1,2\}$, channel input at time~$t$, $\vect{X}_{i,t},$ can depend on its message $M_i$ and the prior output vectors $\vect{Y}_1,\dots,\vect{Y}_{t-1}$ :
\begin{IEEEeqnarray}{rCL}\label{eq:MACinputs}
\vect{X}_{i,t}=\varphi_{i,t}^{(n)} (M_i, \vect{Y}_1,\dots,\vect{Y}_{t-1}),\quad t\in\{1,\dots,n\},
\end{IEEEeqnarray}
for some encoding functions of the form:
\begin{IEEEeqnarray}{rCL}
\varphi_{i,t}^{(n)}&:&\mathcal{M}_i \times \mathbb{R}^{\kappa (t-1)}\to \mathbb{R}.\label{eq:mac_encoding}
\end{IEEEeqnarray}

The channel input sequences $\{\vect{X}_{1,t}\}_{t=1}^n$ and $\{\vect{X}_{2,t}\}_{t=1}^n$ have to satisfy a \emph{total expected average  block-power constraint $P$}:
\begin{IEEEeqnarray}{C}
\label{eq:powerconstraintMAC}
\frac1n \sum_{t=1}^n \left( \mathbf{E}[ \|\vect{X}_{1,t}\|^2]+\mathbf{E}[ \|\vect{X}_{2,t}\|^2]\right) \leq P,
\end{IEEEeqnarray}
where the expectation is over the messages and the realizations of the channel.

The receiver decodes the messages $(M_1,M_2)$ by means of a decoding function $\phi^{(n)}$ of the form 
\begin{IEEEeqnarray}{rCl}
\phi^{(n)}&\colon& \mathbb{R}^{\kappa n}\to     \mathcal{M}_1\times \mathcal{M}_2.\label{eq:mac_decoding} 
\end{IEEEeqnarray}
This means, based on the output sequence $\vect{Y}_1,\ldots, \vect{Y}_n$, the receiver produces its guess 
\begin{equation}\label{eq:estMAC}
(\hat{M}_1,\hat{M}_2)=\phi^{(n)}(\vect{Y}_1,\ldots, \vect{Y}_n). 
\end{equation}

An error occurs in the communication if 
\begin{IEEEeqnarray}{rCl}
(\hat{M}_1,\hat{M}_2) &\neq& (M_1,M_2),
\end{IEEEeqnarray} 
and thus the average probability of error is defined as
\begin{IEEEeqnarray}{rCL} 
 \label{eq:error-mac}
 P_{\e,\MAC}^{(n)}&\triangleq&\textnormal{Pr}\big[ (\hat{M}_1,\hat{M}_2) \neq (M_1,M_2)      \big].
\end{IEEEeqnarray} 

A \emph{$(\lfloor 2^{nR_1}\rfloor,\lfloor 2^{nR_2}\rfloor,n)$ MIMO MAC feedback-code of sum-power $P$} is a triple $$\left(\{\varphi_{1,t}^{(n)}\}_{t=1}^n,\{\varphi_{2,t}^{(n)}\}_{t=1}^n,\Phi^{(n)}\right)$$ 
where $\{\varphi_{1,t}^{(n)}\}_{t=1}^n$ and $\{\varphi_{2,t}^{(n)}\}_{t=1}^n$ are of the form~\eqref{eq:mac_encoding} and satisfy~\eqref{eq:powerconstraintMAC} and  $\phi^{(n)}$ is as in~\eqref{eq:mac_decoding}.

We say that a rate-pair $(R_1,R_2)$ is achievable over the Gaussian MIMO  MAC with feedback under a sum-power constraint $P$, 
if there exists a sequence of $\{(\lfloor 2^{n R_1}\rfloor ,\lfloor 2^{n R_2}\rfloor ,n)\}_{n=1}^\infty$ MIMO MAC 
feedback-codes such that the average probability of a decoding error $P_{\e,\MAC}^{(n)}$ tends to zero as the blocklength $n$ tends to infinity.
The closure of the union of all achievable regions is called \emph{capacity region}. We denote it by $\set{C}_{\MAC}^{\fb}(\trans{\mat{H}_1},\trans{\mat{H}_2},P)$.
The supremum of the sum $R_1+R_2$ over all pairs $(R_1,R_2)$ in $\set{C}_{\MAC}^{\fb}(\trans{\mat{H}_1},\trans{\mat{H}_2},P)$ is called \emph{sum-capacity} and is denoted by $C_{\MAC,\Sigma}^{\fb}(\trans{\mat{H}_1},\trans{\mat{H}_2},P)$.

\subsection{Linear-feedback schemes for MIMO MAC}
In the present paper, we focus on the class of \emph{linear-feedback} coding schemes where the channel inputs at  Transmitter $i$, for $i\in\{1,2\}$, are given by \emph{linear} combinations of the previous feedback signals and an information-carrying vector that only depends on the message $M_i$ (but not on the feedback).

Specifically, we assume that the channel input vectors have the form
\begin{IEEEeqnarray}{rCl}
\vect{X}_{i,t}&=& \vect{W}_{i,t} + \sum_{\tau=1}^{t-1} \mat{C}_{i,\tau,t} \vect{Y}_\tau,\quad i\in\{1,2\},\quad t\in \{1,\dots,n\},
\label{eq:lin_com}
\end{IEEEeqnarray}
where $\vect{W}_{i,t}$ is an information-carrying vector
\begin{IEEEeqnarray}{rCl}
\vect{W}_{i,t}&=&\xi_{i,t}^{(n)}(M_i),\label{eq:wi}
\end{IEEEeqnarray}
and $\{\mat{C}_{i,\tau,t}\}$ are arbitrary $\nu_i$-by-$\kappa$ matrices.

The mappings $\{\xi_{i,t}^{(n)}\colon \set{M}_i \to \Reals^{\nu_i n}\}$ as well as the decoder mapping~$\phi^{(n)}$ can be arbitrary (also non-linear).

The set of all rate-pairs achieved by linear-feedback schemes is called  \emph{linear-feedback capacity region} and is denoted $\set{C}_{\MAC}^{\linfb}(\trans{\mat{H}_1}, \trans{\mat{H}_2};P)$. 
The largest sum-rate achieved by a linear-feedback scheme is called  \emph{linear-feedback sum-capacity} and is denoted $C_{\MAC,\Sigma}^{\linfb}\left(\trans{\mat{H}_1},\trans{\mat{H}_2};P\right)$.

\begin{remark}\label{rem:trans}
For any channel matrices $\trans{\mat{H}}_1$ and $\trans{\mat{H}}_2$ and power constraint $P>0$:
\begin{equation} 
\set{C}_{\MAC}^{\linfb}\left(\trans{\mat{H}_1},\trans{\mat{H}_2};P\right)= \set{C}_{\MAC}^{\linfb}({\mat{\bar H}_1},{\mat{\bar H}_2};P). 
\end{equation}
\end{remark}
\begin{proof} Consider the MIMO MAC with channel matrices $(\trans{\mat{H}_1},\trans{\mat{H}_2})$. If each transmitter multiplies its input vectors by $\mat{E}$ (from the left) before sending the result over the MAC and if the receiver and the transmitters multiply their observed vectors by $\mat{E}$ (from the left) before attempting to decode the messages or before using the feedback, then the MIMO MAC is transformed into a MIMO MAC with channel matrices $({\mat{\bar H}_1},{\mat{\bar H}_2})$. And in the same way the MIMO MAC with channel matrices $({\mat{\bar H}_1},{\mat{\bar H}_2})$ can be transformed into a MIMO MAC with channel matrices $(\trans{\mat{H}_1},\trans{\mat{H}_2})$. This proves the remark.
\end{proof}

\subsection{Previous Results}

Without feedback, the capacity region  of the Gaussian MIMO  MAC under a sum-power constraint $P$, $\set{C}_{\MAC}^{\nofb}\left(\trans{\mat{H}_1},\trans{\mat{H}_2};P\right)$ is readily obtained from the results in~\cite{CHENG-VERDU93}.

With perfect feedback, the capacity region of the MIMO Gaussian MAC under sum-power constraint $P$ is known only in few special cases. An example is the scalar case $\nu_1=\nu_2=\kappa=1$, which we also  call single-input single-output (SISO) setup. In this setup, the channel matrices $(\trans{\mat{H}_1},\trans{\mat{H}_2})$ reduce to the scalar coefficients $(h_1,h_2)$. Ozarow \cite{OZAROW84} determined the capacity region of the scalar Gaussian MAC with perfect feedback under individual power constraints $P_1$ and $P_2$ on the two transmitters' input sequences. It is given by 
\begin{IEEEeqnarray}{rCl}
\set{R}_{\Oz}\left(h_1,h_2;P_1,P_2\right)&=& \bigcup_{\rho\in[0,1]} \set{R}_{\Oz}^\rho(h_1,h_2;P_1,P_2)
\end{IEEEeqnarray}
where for each $\rho\in[0,1]$,  $\set{R}_{\Oz}^\rho(h_1,h_2;P_1,P_2)$ denotes the set of all nonnegative rate-pairs $(R_1,R_2)$ that satisfy
 \begin{subequations}
\begin{IEEEeqnarray}{cCl}
 R_1&\leq& \frac12\log\left(1+h_1^2 P_1 (1-\rho^2)\right),\\
 R_2&\leq&\frac12\log\left(1+h_2^2 P_2 (1-\rho^2)\right),\\
 R_1+R_2&\leq&\hspace*{-0.1cm}\frac12\log\left(1+h_1^2 P_1+h_2^2 P_2+2 \sqrt{h_1^2h_2^2 P_1 P_2}\rho\right).
 \end{IEEEeqnarray}
 \end{subequations} 
From Ozarow's result, we can directly deduce the capacity region of the scalar Gaussian MAC with perfect feedback under a sum-power constraint: 
\begin{IEEEeqnarray}{rCl}
\set{C}_{\MAC,\SISO}^{\fb}\left(h_1, h_2;P\right) & = &\bigcup_{\mathclap{\substack{P_1,P_2\geq 0:\\P_1+P_2=P}}}
\qquad \set{R}_{\Oz}\left(h_1,h_2;P_1,P_2\right).\IEEEeqnarraynumspace\label{eq:mac_capacity}
\end{IEEEeqnarray}
Thus the  capacity region $C_{\MAC,\SISO}^{\fb}$ is achieved by applying Ozarow's scheme with different power splits between the two transmitters. 
The sum-capacity $C_{\MAC,\SISO,\Sigma}^{\fb}\left(h_1, h_2;P\right)$ is 
\begin{IEEEeqnarray}{rCl}\label{sumcapMAC}
C_{\MAC,\SISO,\Sigma}^{\fb}\left(h_1, h_2;P\right)
& = &  \sup_{\mathclap{\substack{P_1,P_2\geq 0:\\P_1+P_2=P}}} \quad \frac{1}{2} \log \Big( 1+ h_1^2 P_1 +h_2^2 P_2 +2 \sqrt{h_1^2 h_2^2P_1P_2}\cdot\rho^\star\left(h_1,h_2;P_1,P_2\right)\Big)\IEEEeqnarraynumspace
 \end{IEEEeqnarray}
where $\rho^\star\left(h_1,h_2;P_1,P_2\right)$ is the unique solution in $[0,1]$ to the following quartic equation in $\rho$
\begin{IEEEeqnarray}{rCl}\label{eq:rhostar}
1+h_1^2 P_1+h_2^2 P_2+2 \sqrt{h_1^2 h_2^2 P_1 P_2} \rho&= &\left(1+h_1^2 P_1 (1- \rho^2)\right) \left(1+h_2^2 P_2 (1- \rho^2)\right). 
\end{IEEEeqnarray}
In Appendix~\ref{app:cor}, we show that in a symmetric setup where  $h_1=h_2=h$, 
\begin{IEEEeqnarray}{rCl}\label{eq:MACsym}
C_{\MAC,\SISO,\Sigma}^{\fb}\left(h,h;P\right) &=&\frac12 \log\left(1+h^2 P(1+\rho^\star(h,h;P/2,P/2))\right).
\end{IEEEeqnarray}
Ozarow's scheme is a \emph{linear feedback scheme} since it combines a Schalkwijk-Kailath~\cite{SK66} type scheme at both transmitters with a no feedback scheme at one of the two transmitters. Specifically, one transmitter sends scaled versions of the linear minimum mean squared estimation  (LMMSE) errors when estimating its message point (which depend only on the message) based on the previous feedback signals. The other transmitter sends the sum of the symbols of a no-feedback scheme and the scaled LMMSE errors about its message point based on the previous feedback signals. Since any no-feedback scheme is a linear-feedback scheme and also the LMMSE errors are by definition linear in the feedback, the overall Ozarow-scheme is also a linear-feedback scheme.
Thus, in the SISO case,
\begin{IEEEeqnarray}{rCl}\label{eq:linopt1}
\set{C}_{\MAC,\SISO}^{\fb}\left(h_1, h_2;P\right) & = & \set{C}_{\MAC,\SISO}^{\linfb}\left(h_1, h_2;P\right),
\end{IEEEeqnarray}
and
\begin{IEEEeqnarray}{rCl}\label{eq:linopt2}
C_{\MAC,\SISO,\Sigma}^{\fb}\left(h_1, h_2;P\right) & = & C_{\MAC,\SISO,\Sigma}^{\linfb}\left(h_1, h_2;P\right).
\end{IEEEeqnarray}

Jafar {\it et al.} \cite{JAFAR06} derived the capacity region with perfect feedback under individual power constraints in the multi-input single-output (MISO) case ($\nu_1, \nu_2$ arbitrary and $\kappa=1$) and in the single-input multi-output (SIMO) case ($\nu_1=\nu_2=1$ and $\kappa$ arbitrary). In both cases the capacity is achieved by a variation of Ozarow's scheme. Based on these results we immediately obtain the linear-feedback capacity region under a total sum-power constraint.
In the MISO case, the channel matrices $\trans{\mat{H}}_1$ and $\trans{\mat{H}}_2$ reduce to the $1\times\nu_1$ and $1\times\nu_2$ vectors $\trans{\vect{h}_1}$ and $\trans{\vect{h}_2}$ and the channel output can be written as
\begin{IEEEeqnarray}{rCl}
\label{eq:MISO_MAC}
Y_t&=&\trans{\vect{h}_1} \vect{x}_{1,t}+\trans{\vect{h}_2} \vect{x}_{2,t}+Z_t.
\end{IEEEeqnarray}
The linear-feedback capacity region is given by
\begin{IEEEeqnarray}{rCl}\label{eq:cap_MISO}
\set{C}_{\MAC,\MISO}^{\linfb}(\trans{\vect{h}_1}, \trans{\vect{h}_2};P)&=& \set{C}_{\MAC,\MISO}^{\fb}(\trans{\vect{h}_1}, \trans{\vect{h}_2};P) \nonumber \\
& = & \set{C}_{\MAC,\SISO}^{\fb}(\|{\vect{h}_1}\|, \|{\vect{h}_2}\|;P),\label{eq:linoptMISO}
\end{IEEEeqnarray}
where notice that the last expression involves the \emph{SISO} capacity region  $\set{C}_{\MAC,\SISO}^{\fb}(\|{\vect{h}_1}\|, \|{\vect{h}_2}\|;P)$.
In the SIMO case, the channel matrices reduce to the $\kappa\times 1$ vectors $\trans{\vect{h}}_1,\trans{\vect{h}}_2$ and the channel output vector can be written as
\begin{IEEEeqnarray}{rCl}
\vect{Y}_t=\vect{h}_1 x_{1,t}+\vect{h}_2 x_{2,t}+\vect{Z}_t.
\end{IEEEeqnarray}
The linear-feedback capacity region  is given by:  
\begin{IEEEeqnarray}{rCl}
\set{C}_{\MAC,\SIMO}^{\linfb}(\trans{\vect{h}_1}, \trans{\vect{h}_2};P) &=& \set{C}_{\MAC,\SIMO}^{\fb}(\trans{\vect{h}_1}, \trans{\vect{h}_2};P) \label{eq:linoptSIMO} \\
& = &\bigcup_{\mathclap{\substack{P_1,P_2\geq 0:\\P_1+P_2=P}}}  \;
\textnormal{cl}\left( \bigcup_{\rho\in[0,1]} \set{R}_{\Jafar}^\rho(\trans{\vect{h}}_1,\trans{\vect{h}}_2;P_1,P_2)\right)\label{eq:exp}\IEEEeqnarraynumspace
\end{IEEEeqnarray}
where for each $\rho\in[0,1]$,  $\set{R}_{\Jafar}^\rho(\trans{\vect{h}}_1,\trans{\vect{h}}_2;P_1,P_2)$ denotes the set of all nonnegative rate-pairs $(R_1,R_2)$ that satisfy
 \begin{subequations}
\begin{IEEEeqnarray}{cCll}
 R_1&\leq& \frac12&\log\left(1+ \|\vect{h}_1\|^2 P_1 (1-\rho^2)\right),\\
 R_2&\leq&\frac12&\log\left(1+ \|\vect{h}_2\|^2 P_2 (1-\rho^2)\right),\\
 R_1+R_2&\leq&\frac12&\log\bigr(1+  \|\vect{h}_1\|^2 P_1+  \|\vect{h}_2\|^2P_2\nonumber\\
 &&&+2 \rho \beta \sqrt{ \|\vect{h}_1\|^2  \|\vect{h}_2\|^2P_1 P_2} \nonumber\\&&&+  \|\vect{h}_1\|^2  \|\vect{h}_2\|^2P_1 P_2 (1-\rho^2)(1-\beta^2)\bigr).\IEEEeqnarraynumspace
 \end{IEEEeqnarray}
 \end{subequations}

\section{Main Results}
\label{sec:main}
\subsection{ Main Results: MAC-BC Duality with Linear-Feedback}
\begin{theorem}
\label{thm:dual}
 \begin{equation}  
  \set{C}_{\BC}^{\linfb}\left(\mat{H}_1,\mat{H}_2;P\right)
= \set{C}_{\MAC}^{\linfb}\left(\trans{\mat{H}_1},\trans{\mat{H}_2};P\right).\end{equation}
\end{theorem}
\begin{proof}
Follows by Propositions~\ref{prop:bc},~\ref{prop:mac}, and~\ref{prop:equal_regions} ahead,  by point~2 of Note~\ref{note:bar}, and because the capacity regions of the MACs with channel matrices $\trans{\mat{H}}_1$ and $\trans{\mat{H}}_2$ and  $\bar{\mat{H}}_1$ and $\bar{\mat{H}}_2$ coincide, see Remark~\ref{rem:trans}. \end{proof}

Theorem~\ref{thm:dual} implies the following corollary on the sum-capacities:
\begin{corollary}
\label{cor:sumrate-MIMO}
 \begin{IEEEeqnarray}{rCl}
  C_{\BC,\Sigma}^{\linfb}\left(\mat{H}_1,\mat{H}_2; P\right) &=&C_{\MAC,\Sigma}^{\linfb}\left(\trans{\mat{H}_1},\trans{\mat{H}_2}; P\right).  \end{IEEEeqnarray}
\end{corollary}

In the scalar case, Theorem~\ref{thm:dual} and Corollary~\ref{cor:sumrate-MIMO} combined with~\eqref{eq:linopt1} and \eqref{eq:linopt2} specialize to: 
\begin{corollary}
\label{thm:dual_scalar}
\begin{IEEEeqnarray}{rCl}  
  \set{C}_{\BC,\SISO}^{\linfb}\left(h_1,h_2;P\right) &=&\set{C}_{\MAC,\SISO}^{\linfb}\left(h_1,h_2;P\right)\\
  &=& \set{C}_{\MAC,\SISO}^{\fb}\left(h_1,h_2;P\right)
\end{IEEEeqnarray}
and 
 \begin{IEEEeqnarray}{rCl}  \label{eq:SISOSum}
  {C}_{\BC,\SISO,\Sigma}^{\linfb}\left(h_1,h_2;P\right)&=&{C}_{\MAC,\SISO,\Sigma}^{\linfb}\left(h_1,h_2;P\right)\\
  &=&{C}_{\MAC,\SISO,\Sigma}^{\fb}\left(h_1,h_2;P\right).
\end{IEEEeqnarray} 
\end{corollary}

 \begin{figure}[ht]
 \centerline{
  \includegraphics[width=8cm]{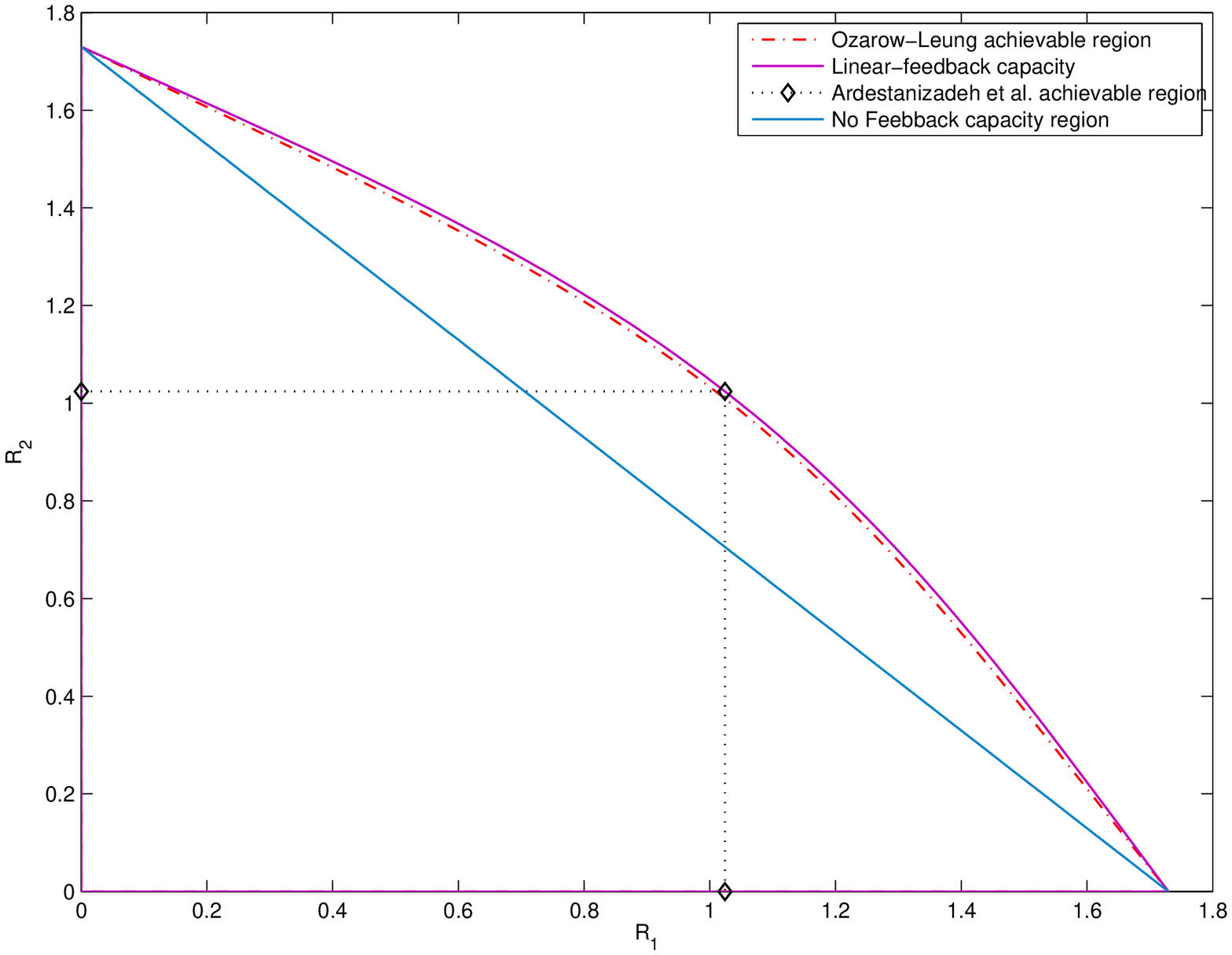}}
  \vspace*{-0.5cm}
  \caption{Achievable regions for the symmetric SISO Gaussian BC  with perfect feedback, with channel coefficients $h_1=h_2=1$ and power constraint $P=10$.\label{fig:sym_bc}}
 \end{figure}

 \begin{figure}[ht]
 \centerline{
  \includegraphics[width=8cm]{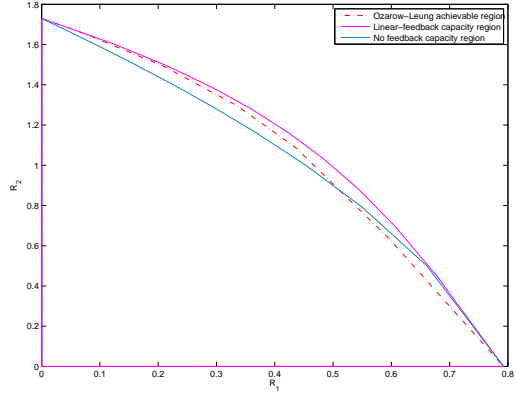}}
  \vspace*{-0.5cm}
  \caption{Achievable regions for the non-symmetric SISO Gaussian BC with perfect feedback, with channel coefficients $h_1=\frac{1}{\sqrt{5}},h_2=1$ and power constraint $P=10$\label{fig:non_sym_bc}.}
 \end{figure}
 Figures~\ref{fig:sym_bc} and~\ref{fig:non_sym_bc} compare the linear-feedback capacity region for the SISO Gaussian BC to the nofeedback capacity region \cite{COVER72,BERGMANS} and to Ozarow \& Leung's achievable region~\cite{OZAROW-LEUNG}.

Using also~\eqref{eq:MACsym}, in the symmetric case we obtain:
\begin{corollary}
\label{cor:sum-rate_BC}
If $h_1=h_2=h$, then 
\begin{IEEEeqnarray}{rCl}
C_{\BC,\SISO,\Sigma}^{\linfb}(h,h; P)
&=& \frac{1}{2} \log \left( 1+ h^2 P +h^2P\cdot\rho^\star(h,h;{P}/{2},{P}/{2})\right),\IEEEeqnarraynumspace\label{eq:maxsymsumrate}
\end{IEEEeqnarray}
where recall that $\rho^\star\left(h_1,h_2;P_1,P_2\right)$ is defined as the solution to the quartic equation in~\eqref{eq:rhostar}.
\end{corollary}
The achievability of the sum-rate in~\eqref{eq:maxsymsumrate} was already established by the control-theory-inspired scheme in \cite{AMF12}. Our result shows thus that for the symmetric scalar Gaussian BC the scheme in  \cite{Elia04}, \cite{Wu05}, \cite{AMF12} is indeed sum-rate optimal among all linear-feedback coding schemes. 

In the SIMO and the MISO case, Theorem~\ref{thm:dual} combined with~\eqref{eq:linoptMISO} and \eqref{eq:linoptSIMO} specialize to: 
 \begin{corollary} Consider the SIMO and MISO cases where the channel matrices reduce to vectors. Let  $\vect{h}_1$ and $\vect{h}_2$ be $\kappa$-dimensional row-vectors. Then,
   \begin{IEEEeqnarray}{rCl}  
    \set{C}_{\BC,\MISO}^{\linfb}({\vect{h}_1},{\vect{h}_2};P)&=&\set{C}_{\MAC,\SIMO}^{\fb}(\trans{\vect{h}}_1,\trans{\vect{h}}_2;P)
\end{IEEEeqnarray} 
Let now $\vect{h}_1$ and $\vect{h}_2$ be $\nu_1$ and $\nu_2$-dimensional column-vectors. Then, 
\begin{IEEEeqnarray}{rCl}
 \set{C}_{\BC,\SIMO}^{\linfb}(\vect{h}_1,\vect{h}_2;P)&=&\set{C}_{\MAC,\MISO}^{\fb}(\trans{\vect{h}_1},\trans{\vect{h}_2};P)\\
  &=& \set{C}_{\MAC,\SISO}^{\fb}(\|\vect{h}_1\|,\|\vect{h}_2\|;P).
  \end{IEEEeqnarray} 
See~\eqref{eq:mac_capacity},  \eqref{eq:linoptMISO},  and \eqref{eq:exp} for computable single-letter characterizations of $\set{C}_{\MAC,\SISO}^{\fb}$, $\set{C}_{\MAC,\SIMO}^{\fb}$, and $\set{C}_{\MAC,\MISO}^{\fb}$. 
  \end{corollary}
%

\subsection{Linear-Feedback Capacity-Achieving Schemes for MAC and BC}\label{sec:dualcoding}
We first describe a class of linear-feedback coding schemes for the BC and the MAC that can achieve  the linear-feedback capacity regions $\set{C}_{\BC}^{\linfb}$ and $\set{C}_{\MAC}^{\linfb}$. 
This allows us to find multi-letter expressions for these capacity regions. 
We then identify pairs of linear-feedback schemes for the BC and the MAC that are dual in the sense that they achieve the same rate-regions. 

The idea of our schemes is to divide the blocklength $n$ into   subblocks of equal length $\eta$ ($\eta$ is a design parameter of our schemes) and to apply an inner code that uses the feedback to transform each subblock of $\eta$ channel uses of the original MIMO BC or MAC into a single channel use of a new MIMO BC or MAC with more transmit and receive antennas. An outer code 
is then applied to communicate 
over the new MIMO BC or MAC without using the feedback. 

We now explain this class of schemes in more detail.

\subsubsection{A class of linear-feedback schemes for the BC}\label{sec:schemeBC}
Fix the blocklength $n$.  The schemes in our class are characterized by the following parameters: 
\begin{itemize}
\item a positive integer $\eta$; 
\item $\kappa$-by-$\nu_1$ matrices $\{\mat{A}_{1,\tau, \ell} \}$, for $\ell=2,\ldots, \eta$ and $\tau=1,\ldots, \ell-1$;   
\item    $\kappa$-by-$\nu_2$ matrices $\{\mat{A}_{2, \tau,\ell} \}$, for $\ell=2,\ldots, \eta$ and $\tau=1,\ldots, \ell-1$;
\item an  encoding mapping $f^{(n')} \colon \set{M}_1 \times \set{M}_2 \to \set{R}^{(\kappa \eta) n'}$ that produces $n'\triangleq \lfloor \frac{n}{\eta}\rfloor$ codevectors (column-vectors) of size $\kappa\eta$
 and 
\item two decoding mappings $g_1^{(n')}\colon \Reals^{(\nu_1 \eta) n'} \to \set{M}_1$ and $g_2^{(n')}\colon \Reals^{(\nu_2 \eta) n'} \to \set{M}_2$ that each decode a block of $n'$ output vectors (column-vectors) of size $\nu_1 \eta$ and $\nu_2 \eta$.
\end{itemize}
As already mentioned, the parameter $\eta$ characterizes the length of the subblocks in our scheme. That means, in our scheme the total blocklength $n$ is divided into  $n'$ subblocks of equal length $\eta$.\footnote{For general blocklength $n$ there will  be a few spare channel uses at the end of each block which we do ignore in our schemes. Since throughout we are interested in the performance limits as $n\to\infty$,  this technicality does not influence our results and will therefore be ignored in the sequel.} The matrices $\{\mat{A}_{1,\tau, \ell} \}$ and  $\{\mat{A}_{2, \tau,\ell} \}$ describe the inner code that is used  within each of the $n'$ subblocks of length $\eta$. Finally, the parameters $f^{(n')}, g_1^{(n')}, g_2^{(n')}$ describe the outer code that is applied to code over the $n'$ subblocks without using the feedback. 

Before describing how the inner code works and how we should choose the encoding and decoding functions of the outer code, we need some definitions. 
Let 
\begin{IEEEeqnarray}{rCl}
\vect{X}&\triangleq& \trans{\begin{pmatrix}
                         \trans{\vect{X}_1},&\dots,&\trans{\vect{X}_\eta}
                        \end{pmatrix}},
\end{IEEEeqnarray}
denote the $\eta\kappa$-dimensional column-vector that is obtained by stacking the first $\eta$  channel input vectors $\vect{X}_1,\ldots, \vect{X}_\eta$ (which are all $\kappa$-dimensional column-vectors) on top of each other. Similarly, for $i\in\{1,2\}$, let 
\begin{IEEEeqnarray}{rCl}      
 \vect{Z}_i&\triangleq& \trans{\begin{pmatrix}
                         \trans{\vect{Z}_{i,1}},&\dots,&\trans{\vect{Z}_{i,\eta}}
                        \end{pmatrix}},       \\  
 \vect{Y}_i&\triangleq& \trans{\begin{pmatrix}
                         \trans{\vect{Y}_{1}},&\dots,&\trans{\vect{Y}_{\eta}}
                       \end{pmatrix}},
\end{IEEEeqnarray}
denote the $\eta \nu_i$ dimensional column-vectors that are obtained by stacking the first $\eta$ noise vectors $\vect{Z}_{1,i},\ldots, \vect{Z}_{i,\eta}$ or channel output vectors $\vect{Y}_{1,i},\ldots, \vect{Y}_{i,\eta}$ on top of each other. Define for $i\in\{1,2\}$, the channel matrices of the $\eta$-length subblocks:
\begin{equation}\label{eq:blockmat}
\BH_i\triangleq \mat{I}_\eta \otimes \mat{H}_i.
\end{equation}
The input-output relation for the first block of $\eta$ channel uses is then summarized as
\begin{equation}\label{eq:input_output}
\vect{Y}_i = \BH_i \vect{X} + \vect{Z}_i, \qquad i\in\{1,2\}.
\end{equation}

Let $\vect{U}$ denote the $\eta \kappa$-dimensional vector produced by outer encoder $f^{(n')}$ for this first block, and define, for $i\in\{1,2\}$, the $\eta\kappa$-by-$\eta\nu_i$ strictly-lower block-triangular matrix
\begin{IEEEeqnarray}{rCl}\label{eq:AB}
\mat{A}_i^{\B}&=&\begin{bmatrix}   
\mat{0}&& \dots &&\mat{0} \\ 
\mat{A}_{i,1,2} &\mat{0}  \\
\mat{A}_{i,1,3}&\mat{A}_{i,2,3} &\mat{0} \\ 
\vdots&&&\ddots\\
\mat{A}_{i,1,\eta}&\mat{A}_{i,2,\eta}&\dots&\mat{A}_{i,(\eta-1),\eta}&\mat{0}
\end{bmatrix},
\end{IEEEeqnarray}
where here $\mat{0}$ denotes the $\kappa$-by-$\nu_i$ matrix with all zero entries.

We now describe how the inner code---specified by the matrices $\{\mat{A}_{1,\tau, \ell} \}$ and $\{\mat{A}_{2,\tau, \ell} \}$--- transforms the first block of $\eta$ channel uses of our original MIMO Gaussian BC into a single channel use of the new MIMO BC. All the other blocks are transformed in a similar way. 
In our scheme, we choose the encoder to produce the following $\eta$ channel inputs in the first block: 
\begin{IEEEeqnarray}{rCl}
\label{eq:bc_inputs}
\vect{X}& = &\left( \mat{I}- \mat{A}_1^\B\BHone-  \mat{A}_2^{\B}\BHtwo\right) \vect{U} +\mat{A}_1^{\B} \vect{Y}_1+ \mat{A}_2^{\B} \vect{Y}_2.
\end{IEEEeqnarray}
(The reason for precoding the codeword vector $\vect{U}$ by the matrix $\left( \mat{I}- \mat{A}_1^\B\BHone-  \mat{A}_2^{\B}\BHtwo\right)$ will become clearer shortly, see~\eqref{eq:echteinputs}.)
By~\eqref{eq:input_output}, the inputs  can also be written as 
\begin{IEEEeqnarray}{rCl}\label{eq:bc_inputs77}
\vect{X}&= & \left( \mat{I}- \mat{A}_1^\B\BHone-  \mat{A}_2^{\B}\BHtwo\right) \vect{U} +\mat{A}_1^{\B}(\BHone \vect{X}+ \vect{Z}_1)+ \mat{A}_2^{\B} (\BHone\vect{X}+\vect{Z}_2)
\end{IEEEeqnarray}
and thus, 
\begin{IEEEeqnarray}{rCl}\label{eq:bc_inputs771}
 \left( \mat{I}- \mat{A}_1^\B\BHone-  \mat{A}_2^{\B}\BHtwo\right) \vect{X} &= & \left( \mat{I}- \mat{A}_1^\B\BHone-  \mat{A}_2^{\B}\BHtwo\right) \vect{U} +\mat{A}_1^{\B} \vect{Z}_1+ \mat{A}_2^{\B} \vect{Z}_2.
\end{IEEEeqnarray}
Multiplying both sides of~\eqref{eq:bc_inputs771} from the left by the invertible matrix $ \left( \mat{I}- \mat{A}_1^\B\BHone-  \mat{A}_2^{\B}\BHtwo\right)^{-1}$ results in:
\begin{IEEEeqnarray}{rCl}\label{eq:echteinputs}
\vect{X}&= & \vect{U} + \mat{B}_1^{\B} \vect{Z}_1 + \mat{B}_2^{\B} \vect{Z}_2,
\end{IEEEeqnarray}
where we defined 
\begin{equation}\label{eq:B}
\mat{B}_i^{\B}\triangleq\left( \mat{I}- \mat{A}_1^\B\BHone-  \mat{A}_2^{\B}\BHtwo\right)^{-1} \mat{A}_i^{\B}, \quad i\in \{1,2\}.
\end{equation}
By~\eqref{eq:input_output} the corresponding outputs can be written as
\begin{subequations}\label{eq:outputsBC}
\begin{IEEEeqnarray}{rCl}
\label{eq:output1}
\vect{Y}_1& = \BHone \vect{U} + (\mat{I}+\BHone\mat{B}_1^{\B}) \vect{Z}_1 + \BHone \mat{B}_2^{\B} \vect{Z}_2,\\
\vect{Y}_2& = \BHtwo \vect{U} + (\mat{I}+\BHtwo\mat{B}_2^{\B} )\vect{Z}_2 + \BHtwo \mat{B}_1^{\B} \vect{Z}_1.\label{eq:output2}
\end{IEEEeqnarray}
\end{subequations}

Inspecting~\eqref{eq:echteinputs}, we see that the  channel inputs $\{\vect{X}_t\}_{t=1}^n$ to our original MIMO BC satisfy the average block-power constraint~\eqref{eq:powerconstraintBC} if 
\begin{equation}
\textnormal{tr}\left(\mat{ B}_1^{\B} \trans{ \mat{(\mat{ B}_1^{\B})}}\right)+ \textnormal{tr}\left(\mat{B}_2^{\B} \trans{\mat{(\mat{ B}_2^{\B})}}\right) \leq \eta P
\end{equation} and if the $n'$ codevectors 
produced by the outer encoder $f^{(n')}$ are average block-power constrained to  power
\begin{equation}\label{eq:bc_mimo_pc}
\eta P- \textnormal{tr}\left(\mat{ B}_1^{\B} \trans{ \mat{(\mat{ B}_1^{\B})}}\right)- \textnormal{tr}\left(\mat{B}_2^{\B} \trans{\mat{(\mat{ B}_2^{\B})}}\right).
\end{equation}

\begin{definition}\label{def:RBC}Let $\set{R}_{\BC}\left(\eta, \mat{B}_1^{\B}, \mat{B}_2^{\B},\BHone, \BHtwo;P\right)$ denote the capacity region of the MIMO Gaussian BC 
in~\eqref{eq:outputsBC} \emph{without feedback} when the vector-input $\vect{U}$ is average block-power constrained to~\eqref{eq:bc_mimo_pc}.
\end{definition}

The outer code $\{f^{(n')},g_1^{(n')}, g_2^{(n')}\}$ is designed to achieve the nofeedback capacity of the new MIMO Gaussian BC  in~\eqref{eq:outputsBC}  under average input-power constraint $\eta P- \textnormal{tr}(\mat{ B}_1^{\B} \trans{ \mat{(\mat{ B}_1^{\B})}})- \textnormal{tr}(\mat{B}_2^{\B} \trans{\mat{(\mat{ B}_2^{\B})}})$.

Combining all this, we conclude that over the original MIMO Gaussian BC with feedback our overall scheme (consisting of inner and outer code) achieves the rate region $\set{R}_{\BC}\left(\eta, \mat{B}_1^{\B}, \mat{B}_2^{\B},\BHone, \BHtwo;P\right)$ scaled by a factor $\frac{1}{\eta}$. 
In view of the following Note~\ref{lm1}, it thus follows that our schemes achieve the rate region in~\eqref{eq:reg_bc} ahead.

\begin{note}\label{lm1} Let $\set{T}\triangleq\set{T}_1\times \set{T}_2$ 
where $\set{T}_i$, for $i\in\{1,2\}$, denotes the set of strictly-lower block-triangular matrices with block matrices of size $\kappa\times\nu_i$. The mapping described by~\eqref{eq:B} has the form
\begin{IEEEeqnarray}{rCCC}
\omega \colon &\set{T} &\to &\set{T} \nonumber \\
&(\mat{A}_1^\B, \mat{A}_2^\B)& \mapsto &(\mat{B}_1^{\B}, \mat{B}_2^{\B}), 
\end{IEEEeqnarray}
and is bijective.
\end{note}

\begin{proof}
See Appendix~\ref{sec:pflm1}.
\end{proof}

\begin{proposition}
\label{prop:bc}
The linear-feedback capacity region of the MIMO Gaussian BC with channel matrices $\mat{H}_1$ and $\mat{H}_2$ under a sum-power constraint $P$ is:
\begin{IEEEeqnarray}{rCl}
\set{C}_{\BC}^{\linfb}\left(\mat{H}_1,\mat{H}_2;P\right)&=&\textnormal{cl}\left(\bigcup_{\eta,\mat{B}_1^{\B},\mat{B}_2^{\B}}\frac{1}{\eta}\set{R}_{\BC}\left(\eta,\mat{ B}_1^{\B},\mat{B}_2^{\B},\mat{H}_1^\B,\mat{H}_2^\B;P\right)\right)\label{eq:reg_bc}
\end{IEEEeqnarray}
where the union is over all positive integers $\eta$ and all strictly-lower block-triangular $(\eta \kappa)$-by-$(\eta \nu_1)$ and $(\eta \kappa)$-by-$(\eta \nu_2)$  matrices $\mat{B}_1^{\B}$ and $\mat{B}_2^{\B}$ with blocks of sizes $\kappa\times\nu_1$ and $\kappa\times \nu_2$ that satisfy 
\begin{equation}
\textnormal{tr}\left(\mat{B}_1^{\B} \trans{\mat{(\mat{B}_1^{\B})}}\right)+\textnormal{tr}\left(\mat{B}_2^{\B} \trans{\mat{(\mat{B}_2^{\B})}}\right)\leq \eta P.
\end{equation}
\end{proposition}
\begin{proof} 
We already concluded the achievability part (see the paragraphs preceding Note~\ref{lm1}).
The converse is proved in Section~\ref{sec:con_prop_bc}.
\end{proof}

\subsubsection{A class of linear-feedback schemes for the MAC}\label{sec:schemeMAC}
We fix the blocklength $n$. The schemes in our class are parametrized by 
\begin{itemize}
\item a positive integer $\eta$; 
\item $\nu_1$-by-$\kappa$ matrices $\{\mat{C}_{1,\tau,\ell} \}$, for $\ell=2,\ldots, \eta$ and $\tau=1,\ldots, \ell-1$;  
\item    $\nu_2$-by-$\kappa$ matrices $\{\mat{C}_{2,\tau, \ell} \}$, for $\ell=2,\ldots, \eta$ and $\tau=1,\ldots, \ell-1$;
\item two  encoding mappings $f_1^{(n')} \colon \set{M}_1 \to \set{R}^{(\nu_1 \eta) n'}$ and  $f_2^{(n')} \colon \set{M}_2 \to \set{R}^{(\nu_2 \eta) n'}$ that produce $n'\triangleq \lfloor \frac{n}{\eta}\rfloor$ codevectors (column-vectors) of sizes $\nu_1\eta$ and $\nu_2\eta$, respectively; and 
\item a decoding mapping $g^{(n')}\colon \Reals^{(\kappa \eta) n'} \to \set{M}_1\times \set{M}_2$ that decodes a block of $n'$ output vectors (column-vectors) of length $\kappa \eta$. 
\end{itemize}
 Similar to the BC schemes, the parameter $\eta$ characterizes the length of the subblocks in our scheme. That means, the total blocklength $n$ is again divided into  $n'$ subblocks of equal length $\eta$. The matrices $\{\mat{C}_{1,\tau, \ell} \}$ and  $\{\mat{C}_{2, \tau,\ell} \}$ describe the inner code that is used  within each of the $n'$ subblocks of length $\eta$. Finally, the parameters $f_1^{(n')}, f_2^{(n')}, g^{(n')}$ describe the outer code that is applied to code over the $n'$ subblocks without using the feedback.

Before describing how the inner code works and how to design the outer code, we need to introduce some notation. 
Let, for $i\in\{1,2\}$, 
\begin{IEEEeqnarray}{rCl}
\vect{X}_i&\triangleq & \trans{\begin{pmatrix}                        \trans{\vect{X}_{i,1}},&\dots,&\trans{\vect{X}_{i,\eta}}
\end{pmatrix}},
\end{IEEEeqnarray}
denote the $\eta \nu_i$-dimensional column-vector that is obtained by stacking the first $\eta$  channel input vectors $\vect{X}_{i,1},\ldots, \vect{X}_{i,\eta}$ (which are all $\nu_i$-dimensional column-vectors) on top of each other. Similarly, let 
\begin{IEEEeqnarray}{rCl}
\vect{Y}&\triangleq& \trans{\begin{pmatrix}
                         \trans{\vect{Y}_1},&\dots,&\trans{\vect{Y}_\eta}
                        \end{pmatrix}}\\
\vect{Z}&\triangleq& \trans{\begin{pmatrix}
                         \trans{\vect{Z}_1},&\dots,&\trans{\vect{Z}_\eta}
                        \end{pmatrix}}                     
\end{IEEEeqnarray}
denote the $\eta \kappa$-dimensional column vectors that are obtained by stacking the first $\eta$  noise vectors $\vect{Z}_1,\dots, \vect{Z}_\eta$ and channel output vectors $\vect{Y}_1,\dots, \vect{Y}_\eta$ on top of each other. Using the definition of the block channel matrices in~\eqref{eq:blockmat}, we can summarize the input-output relation for the first block of $\eta$ channel uses as
\begin{equation}\label{eq:echtinputoutputMAC}
\vect{Y}= \tBHone \vect{X}_1 + \tBHtwo \vect{X}_2 + \vect{Z}. 
\end{equation}

Let $\vect{U}_1$ and $\vect{U}_2$ denote the $\eta\nu_1$ and  $\eta\nu_2$-length codevectors (column-vectors) produced by $f_1^{(n')}$ and $f_2^{(n')}$ for this first block, 
and define the strictly-lower block-triangular matrices
\begin{IEEEeqnarray}{rCl}\label{eq:CB}
\mat{C}_i^\B&=&\begin{bmatrix}   
\mat{0}&& \dots &&\mat{0} \\ 
\mat{C}_{i,1,2} &\mat{0}  \\
\mat{C}_{i,1,3}&\mat{C}_{i,2,3} &\mat{0} \\ 
\vdots&&&\ddots\\
\mat{C}_{i,1,\eta}&\mat{C}_{i,2,\eta}&\dots&\mat{C}_{i,(\eta-1),\eta}&\mat{0}
\end{bmatrix},\quad i\in\{1,2\},
\end{IEEEeqnarray}
where here $\mat{0}$ denotes an $\nu_i$-by-$\kappa$ zero matrix.  
Also, let 
\begin{IEEEeqnarray}{rCl}\label{eq:CD}
 \mat{D}_i^{\B}&\triangleq&\mat{C}_i^{\B}\left(\mat{I}-\tBHone \mat{C}^{\B}_1-\tBHtwo \mat{C}_2^{\B}\right)^{-1},~i\in \{1,2\},
\end{IEEEeqnarray} and let $\mat{Q}_1$  be the unique positive square root of the (positive-definite)   $\nu_1\eta$-by-$\nu_1\eta$  matrix
\begin{subequations}
\label{eq:Mi}
\begin{IEEEeqnarray}{rCl}
\mat{M}_1&\triangleq&\trans{(\mat{I}+\mat{D}_1^{\B}\tBHone)}(\mat{I}+ \mat{D}_1^{\B}\tBHone)+\trans{(\mat{D}_2^{\B}\tBHone)}(\mat{D}_2^{\B} \tBHone)
\label{eq:M1}
 \end{IEEEeqnarray}
 and  $\mat{Q}_2$ be the unique positive square root of the (positive-definite)   $\nu_2\eta$-by-$\nu_2\eta$  matrix
 \begin{IEEEeqnarray}{rCl}
 \mat{M}_2&\triangleq&\trans{(\mat{I}+\mat{D}_2^{\B}\tBHtwo)} (\mat{I}+ \mat{D}_2^{\B}\tBHtwo)+\trans{(\mat{D}_1^{\B} \tBHtwo)}(\mat{D}_1^{\B}\tBHtwo).
 \end{IEEEeqnarray}
\end{subequations}

We can now describe how the inner code---specified by the matrices $\{\mat{C}_{1,\tau,\ell} \}$ and $\{\mat{C}_{2,\tau,\ell} \}$---transforms the first block of $\eta$ channel uses into a single channel use of the new MIMO MAC. The transformation of the other blocks is done in a similar way.
Transmitter $i$'s, $i\in\{1,2\}$,  $\eta$ inputs in the first block are 
\begin{IEEEeqnarray}{rCl}
\vect{X}_i&=&\mat{Q}_i^{-1} \vect{U}_{i} +\mat{C}_i^{\B} \vect{Y}\label{eq:input_mac}.
\end{IEEEeqnarray}
Thus, by~\eqref{eq:echtinputoutputMAC}, the corresponding outputs $\vect{Y}$ satisfy
\begin{IEEEeqnarray}{rCl}\label{eq:iiii}
\vect{Y}&=&\tBHone \mat{Q}_1^{-1} \vect{U}_{1}   +\tBHtwo \mat{Q}_2^{-1} \vect{U}_{2} +\left(\tBHone \mat{C}_1^{\B} + \tBHtwo \mat{C}_2^{\B}\right) \vect{Y} + \vect{Z}\
\end{IEEEeqnarray}
Subtracting $(\tBHone \mat{C}_1^{\B} + \tBHtwo \mat{C}_2^{\B}) \vect{Y}$ from both sides of~\eqref{eq:iiii} and then multiplying both sides from the left by the matrix $( \mat{I} - \tBHone \mat{C}_1^{\B} - \tBHtwo \mat{C}_2^{\B})^{-1}$, we  obtain
\begin{IEEEeqnarray}{rCl}\label{eq:blockoutputs}
\vect{Y}&=&( \mat{I} - \tBHone \mat{C}_1^{\B} - \tBHtwo \mat{C}_2^{\B})^{-1} \cdot \left( \tBHone \mat{Q}_1^{-1} \vect{U}_{1}   +\tBHtwo \mat{Q}_2^{-1} \vect{U}_{2} + \vect{Z}\right). \label{eq:y_mac1}\end{IEEEeqnarray}
In view of the definition in \eqref{eq:CD}, the inputs  in~\eqref{eq:input_mac}  satisfy
\begin{IEEEeqnarray}{rCl}
\vect{X}_i
&=&\mat{Q}_i^{-1} \vect{U}_i+\mat{D}_i^{\B}  \bigr(\tBHone \mat{Q}_1^{-1} \vect{U}_1 + \tBHtwo \mat{Q}_2^{-1}\vect{U}_2+ \vect{Z}\bigr). \label{eq:blockinputs}\IEEEeqnarraynumspace
\end{IEEEeqnarray}

\begin{lemma}\label{lem1}
In our scheme, the channel inputs $\{\vect{X}_{1,t}\}_{t=1}^n$ and $\{\vect{X}_{2,t}\}_{t=1}^n$ to the original MIMO Gaussian MAC satisfy the total average block-power constraint~\eqref{eq:powerconstraintMAC} whenever  
\begin{equation}
\textnormal{tr}\left(\mat{D}_1^{\B} \trans{\mat{(\mat{D}_1^{\B})}}\right)+ \textnormal{tr}\left(\mat{D}_2^{\B} \trans{\mat{(\mat{D}_2^{\B})}}\right) \leq \eta P
\end{equation} 
and the codevectors produced by $f_{1}^{(n')}$ and $f_{2}^{(n')}$ are total average block-power constrained to power
\begin{equation}\label{eq:new_mac_pc}
\eta P- \textnormal{tr}\left(\mat{D}_1^{\B} \trans{\mat{(\mat{D}_1^{\B})}}\right)- \textnormal{tr}\left(\mat{D}_2^{\B} \trans{\mat{(\mat{D}_2^{\B})}}\right).
\end{equation}
\end{lemma}
\begin{proof}
See Section~\ref{sec:pflm2}.
\end{proof} 
\begin{definition}\label{def:RMAC}Let $\set{R}_{\MAC}\left(\eta, \mat{D}_1^{\B}, \mat{D}_2^{\B},\tBHone, \tBHtwo;P\right)$ denote the capacity region of the MIMO Gaussian MAC \emph{without feedback} in~\eqref{eq:y_mac1} under average block-power constraint
~\eqref{eq:new_mac_pc} on the input vectors $\vect{U}_1$ and $\vect{U}_2$.
 \end{definition}

The outer code $\{f_1^{(n')}, f_2^{(n')}, g^{(n')}\}$ is designed so that it achieves the nofeedback capacity of the new MIMO Gaussian MAC  in~\eqref{eq:blockinputs} under average input-power constraint $\eta P-  \textnormal{tr}(\mat{D}_1^{\B} \trans{\mat{(\mat{D}_1^{\B})}})- \textnormal{tr}(\mat{D}_2^{\B} \trans{\mat{(\mat{D}_2^{\B})}})$. 

Combining all this, we conclude that over the original MIMO Gaussian MAC our overall scheme (consisting of inner and outer code) achieves  the rate region $\set{R}_{\MAC}\left(\eta, \mat{D}_1^{\B}, \mat{D}_2^{\B},\tBHone, \tBHtwo;P\right)$ scaled by a factor~$\frac{1}{\eta}$.  In view of the following Note~\ref{lm2}, it thus follows that our schemes achieve the rate region in~\eqref{eq:MACprop}. 

\begin{note}\label{lm2} Let $\set{\tilde T}\triangleq \set{\tilde T}_1\times \set{\tilde T}_2$ where, for $i\in\{1,2\}$, $\set{\tilde T}_i$ denotes  the set of strictly-lower block-triangular matrices with block matrices of size $\nu_i\times\kappa$. The mapping described in~\eqref{eq:CD} is of the form  
\begin{IEEEeqnarray}{rCCC}
\tilde \omega \colon &\set{\tilde T} &\to &\set{\tilde T} \nonumber \\
&(\mat{C}_1^\B, \mat{C}_2^\B)& \mapsto &(\mat{D}_1^{\B}, \mat{D}_2^{\B}), 
\end{IEEEeqnarray}
and is bijective.
\end{note}
\begin{proof}
Analogous to the proof of Note~\ref{lm1}. Details omitted.
\end{proof}

\begin{proposition}
 \label{prop:mac}
 The linear-feedback capacity of the  Gaussian MIMO MAC with channel matrices $\trans{\mat{ H}_1}$ and $\trans{\mat{ H}_2}$ under a sum-power constraint $P$ satisfies
\begin{IEEEeqnarray}{rCl}
\set{C}_{\MAC}^{\linfb}\left(\trans{\mat{ H}}_1,\trans{\mat{ H}}_2;P\right)
&=&\textnormal{cl}\left(\bigcup_{\eta,\mat{D}_1^{\B},\mat{D}_2^{\B}}\frac{1}{\eta}\set{R}_{\MAC}\left(\eta,\mat{D}_1^{\B},\mat{D}_2^{\B},\tBHone,\tBHtwo;P\right)\right)\label{eq:MACprop}
\IEEEeqnarraynumspace
\end{IEEEeqnarray}
where the union is over all positive integers $\eta$ and all strictly-lower block-triangular $(\eta \nu_1)$-by-$(\eta \kappa)$ and $(\eta\nu_2)$-by-$(\eta\kappa)$ matrices $\mat{D}_1^{\B}$ and $\mat{D}_2^{\B}$ with blocks of sizes $\nu_1\times\kappa$ and $\nu_2\times\kappa$ that satisfy  
\begin{equation}
\textnormal{tr}\left(\mat{D}_1^{\B} \trans{\mat{(\mat{D}_1^{\B})}}\right)+ \textnormal{tr}\left(\mat{D}_2^{\B} \trans{\mat{(\mat{D}_2^{\B})}}\right)\leq \eta P.
\end{equation}
\end{proposition}
\begin{proof} The achievability follows from the considerations above. 
The converse is proved in Section~\ref{sec:conv_prop_mac}.\end{proof}

\subsubsection{Dual linear-feedback schemes for MAC and BC}
Recall that for any matrix $\mat{M}$, we defined $\mat{\bar{M}}\triangleq \mat{E}\trans{\mat{M}}\mat{E}$, where $\mat{E}$ denotes the exchange matrix with appropriate dimensions.
\begin{proposition}
\label{prop:equal_regions}Let $\bBH_i \eqdef\mat{I}_\eta \otimes\mat{\bar{H}}_i$.
If 
\begin{IEEEeqnarray}{rCl}
\mat{B}_i^{\B}&=&\bar{\mat{D}}_i^{\B},\quad i\in\{1,2\},\label{eq:choice}
\end{IEEEeqnarray}
 then the following two regions coincide: 
\begin{IEEEeqnarray}{rCl}
\set{R}_\BC\left(\eta,\mat{B}_1^{\B},\mat{B}_2^{\B},\mat{ H}_1^\B,\mat{ H}_2^\B;P\right)&=& 
\set{R}_\MAC\left(\eta,\mat{D}_1^{\B},\mat{D}_2^{\B},\bBHone, \bBHtwo;P\right).\label{eq:equal_regions}
\end{IEEEeqnarray}
\end{proposition}
%
%
\begin{proof}
See Section~\ref{sec:proofprop3}.
\end{proof}

When $\{\mat{A}^{\B}_i,\mat{B}^{\B}_i\}_{i=1}^2$ satisfy~\eqref{eq:B} and $\{\mat{C}^{\B}_i,\mat{D}^{\B}_i\}_{i=1}^2$ satisfy~\eqref{eq:CD}, Condition~\eqref{eq:choice} is equivalent to 
\begin{equation}\label{eq:ACone}
\mat{A}^{\B}_i = \mat{\bar{C}}_i^{\B}.
\end{equation}

Combining Proposition~\ref{prop:equal_regions}, Equality~\eqref{eq:ACone}, and Remark~\ref{rem:trans} we obtain:

\begin{corollary}\label{cor:param} Consider a MIMO Gaussian BC with channel matrices $(\mat{H}_1, \mat{H}_2)$ and its dual MAC with channel matrices $(\trans{\mat{H}}_1, \trans{\mat{H}}_2)$.
Fix the MAC-scheme parameters $\eta$, $\{\mat{C}_{1,\tau, \ell} \}$, $\{\mat{C}_{2, \tau,\ell} \}$, and let $f_1^{(n')}$, $f_2^{(n')}$, $g^{(n')}$ be an optimal outer code for these choices. Choose now the BC-scheme parameters
\begin{IEEEeqnarray}{rCl}\label{eq:equivalence_cond}
\mat{A}_{i, \tau,\ell} = \mat{\bar C}_{i,\eta-\tau, \eta-\ell+2}, 
\end{IEEEeqnarray}
and an optimal outer code $f^{(n')}$, $g_1^{(n')}$, and $g_2^{(n')}$ as described in \cite{WSS06}. Then, our MAC and BC-schemes achieve the same rate regions:
\begin{IEEEeqnarray}{rCl}
\set{R}_\BC\left(\eta,\mat{B}_1^{\B},\mat{B}_2^{\B},\mat{ H}_1^\B,\mat{ H}_2^\B;P\right)&=& 
\set{R}_\MAC\left(\eta,\mat{D}_1^{\B},\mat{D}_2^{\B},\tBHone, \tBHtwo;P\right).
\end{IEEEeqnarray}

In the SISO case, all conditions \eqref{eq:equivalence_cond} are summarized by
\begin{IEEEeqnarray}{rCl}\label{eq:AC}
\mat{A}_i^\B = \mat{\bar C}_i^{\B}. 
\IEEEeqnarraynumspace
\end{IEEEeqnarray}
\end{corollary}

\begin{proof}
 See Section~\ref{sec:proofcor5}.
\end{proof}

In view of Corollary~\ref{cor:param} and the capacity-achieving schemes in \cite{OZAROW84} and \cite{JAFAR06}, 
for the SISO, the SIMO, and the MISO MAC, we can readily deduce the  parameters of our linear-feedback schemes in Section~\ref{sec:schemeBC} that achieve the linear-feedback capacity of the dual BCs.

\section{Extension I: One-Sided Feedback}
\label{sec:ext_one}
In this section we assume that there is feedback from only one side. That means, in the BC, there is 
feedback from only one of the two receivers, and in the MAC only one of the two transmitters has feedback.
\subsection{MIMO Gaussian  BC with One-Sided Feedback}
\;
\begin{figure}[ht]
\centerline{
\begin{tikzpicture}[xscale=1,yscale=1]
\draw [thick] (0,2.1) rectangle (1.75,2.9);
\node[right] at (-0.1,2.55){Transmitter};
\node[left] at (0,2.55){$(M_1,M_2)$};
\draw [->][thick](1.75,2.5)--(2.4,3.05);
\draw [thick](2.5,3.1) circle [radius=0.13];
\node[above left] at (2.5,3.1){$\mat{H}_1$};
\node  at (2.5,3.1){$\times$};
\draw [->][thick](2.6,3.2)--(3,3.5)--(3.5,3.5);
\node[right] at (1.85,2.5){$\vect{x}_t$};
\draw [thick](3.63,3.5) circle [radius=0.13];
\node at (3.63,3.5){$+$};
\draw [->][thick](3.63,4)--(3.63,3.65);
\node [left] at (3.68,3.8){$\vect{Z}_{1,t}$};
\draw [->][thick](3.78,3.5)--(4.45,3.5);
\draw [thick] (4.45,3.1) rectangle (6.1,3.9);
\node[right] at (4.4,3.5){Receiver~1};
\node[right] at (6,3.5){$\hat{M}_1$};
\draw [->][thick](1.75,2.5)--(2.4,2);
\draw [thick](2.5,1.9) circle [radius=0.13];
\node[below left] at (2.5,1.8){$\mat{H}_2$};
\node  at (2.5,1.9){$\times$};
\draw [->][thick](2.6,1.83)--(3,1.5)--(3.5,1.5);
\draw [thick](3.63,1.5) circle [radius=0.13];
\node at (3.63,1.5){$+$};
\draw [->][thick](3.63,1)--(3.63,1.35);
\node [left] at (3.68,1.2){$\vect{Z}_{2,t}$};
\draw [->][thick](3.78,1.5)--(4.45,1.5);
\draw [thick] (4.45,1.1) rectangle (6.1,1.9);
\node[right] at (4.4,1.5){Receiver~2};
\node[right] at (6,1.5){$\hat{M}_2$};
\node[below] at (4.1,3.5){$\vect{Y}_{1,t}$};
\node[above] at (4.1,1.5){$\vect{Y}_{2,t}$};
\draw [->][dashed,thick][red](5.25,3.9)--(5.25,4.3)--(0.9,4.3)--(0.9,2.9);
\end{tikzpicture}}
\caption{Two-user MIMO Gaussian BC with one-sided feedback.\label{fig:bc_one}} 
\end{figure}
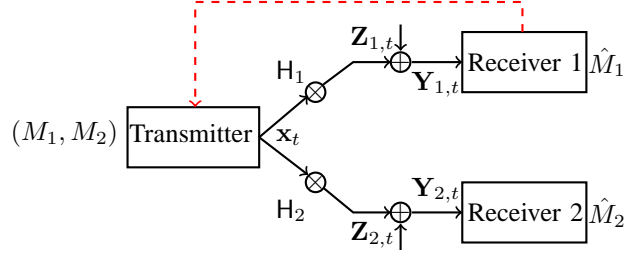

Consider the Gaussian MIMO BC described in~\eqref{eq:BCmodel}, Section~\ref{sec:BC}, but with feedback only from Receiver~1 (see Figure~\ref{fig:bc_one}). The inputs are thus of the form
\begin{IEEEeqnarray}{rCL}\label{eq:inputsBCone}
\vect{X}_t = \varphi_t^{(n)}(M_1, M_2, \vect{Y}_{1,1}, \ldots, \vect{Y}_{1,t-1}),\quad t\in\{1,\dots,n\}.\IEEEeqnarraynumspace
\end{IEEEeqnarray}
We will again restrict to linear-feedback schemes where the inputs are generated as 

\begin{IEEEeqnarray}{rCl}
\vect{X}_t&=&\vect{W}_t+ \sum_{\tau=1}^{t-1} \mat{A}_{1,\tau,t}\vect{Y}_{1,\tau}, \quad t\in \{1,\dots,n\}, \label{eq:lin_com_2_bcone}
\end{IEEEeqnarray}
where $\vect{W}_t= \xi_t^{(n)}(M_1, M_2)$ for arbitrary functions~$\xi_t^{(n)}$ and where $\mat{A}_{1,\tau,t}$ are arbitrary $\kappa$-by-$\nu_1$ matrices.

Decodings, power constraint,  and the definitions of error probabilities and capacity regions are as described in Section~\ref{sec:BC}. 

We denote the linear-feedback capacity region with one-sided feedback from Receiver~1 by $\set{C}_{\BC, \One}^{\linfb}(\mat{H}_1, \mat{H}_2; P)$. It is unknown to date. Inner bounds (i.e., achievable regions) have been proposed by Bhaskaran~\cite{sibi} and Steinberg, Lapidoth, and Wigger~\cite{LSW10}. 

Analogous to Proposition~\ref{prop:bc}, we can derive a multi-letter expression for the linear-feedback capacity region 
$\set{C}_{\BC, \One}^{\linfb}({\mat{H}}_1, {\mat{H}}_2; P)$. Recall the definition of the regions 
$\set{R}_{\BC}$ in Definition~\ref{def:RBC}.
\begin{proposition}\label{prop:bcone}
\begin{IEEEeqnarray*}{rCl}
\set{C}_{\BC,\One}^{\linfb}({\mat{ H}}_1,{\mat{ H}}_2;P)=\textnormal{cl}\left(\bigcup_{\eta,\mat{B}_1^{\B}}\frac{1}{\eta}\set{R}_{\BC}\left(\eta,\mat{B}_1^{\B},\mat{0},\BHone,\BHtwo;P\right)\right)
\IEEEeqnarraynumspace
\end{IEEEeqnarray*}
where the union is over all positive integers $\eta$ and all strictly-lower block-triangular $(\eta \kappa)$-by-$(\eta \nu_1)$ matrices $\mat{B}_1^{\B}$  with blocks of sizes $\kappa\times\nu_1$ that satisfy $\textnormal{tr}\left(\mat{B}_1^{\B} \trans{\mat{(\mat{B}_1^{\B})}}\right)\leq \eta P,$ and where $\mat{0}$ denotes the $(\eta \kappa)$-by-$(\eta \nu_2)$ all-zero matrix.
\end{proposition}
\begin{proof}
Analogous to the proof of Proposition~\ref{prop:bc}, but where the matrix $\mat{B}_2^{\B}$ needs to be the $(\eta \kappa)$-by-$(\eta \nu_2)$ all-zero matrix, which by~\eqref{eq:B} implies that also  $\mat{A}_{2,\tau,\ell}=\mat{0}$ for all $\tau,\ell$.
\end{proof}

\subsection{MIMO Gaussian  MAC with One-Sided Feedback}
\; 
 \begin{figure}[ht]
\centerline{
\begin{tikzpicture}[xscale=1,yscale=1]
\draw [thick] (-0.5,3.1) rectangle (1.5,3.9);
\node[right] at (-0.55,3.5){Transmitter~1};
\node[left] at (-0.5,3.5){$M_1$};
\draw [->][thick](1.5,3.5)--(2.35,3.5);
\draw [->][thick](2.6,3.5)--(3.05,2.6);
\draw [thick](2.45,3.5) circle [radius=0.13];
\node[above right] at (2.5,3.5){$\trans{\mat{H}_1}$};
\node [right] at (2.2,3.5){$\times$};
\node[above] at (2,3.5){$\vect{x}_{1,t}$};
\draw [->][thick](1.5,1.5)--(2.35,1.5);
\draw [->][thick](2.6,1.5)--(3.05,2.44);
\draw [thick](2.45,1.5) circle [radius=0.13];
\node [right] at (2.2,1.5){$\times$};
\node[below right] at (2.5,1.5){$\trans{\mat{H}_2}$};
\node[below] at (2,1.5){$\vect{x}_{2,t}$};
\draw [thick] (-0.5,1.1) rectangle (1.5,1.9);
\node[right] at (-0.55,1.5){Transmitter~2};
\node[left] at (-0.5,1.5){$M_2$};
\draw [thick] (4,2.1) rectangle (5.5,2.9);
\node[right] at (4,2.55){Receiver};
\node[right] at (5.4,2.55){$(\hat{M}_1,\hat{M}_2)$};
\draw [->][thick](3.3,2.5)--(4,2.5);
\node[below] at (3.7,2.5){$\vect{Y}_t$};
\draw [thick](3.15,2.5) circle [radius=0.13];
\node at (3.15,2.5){$+$};
\draw [->][thick](3.15,3.18)--(3.15,2.62);
\node [right] at (3.15,3.02){$\vect{Z}_t$};
\draw [->][dashed,thick][red](4.75,2.9)--(4.75,4.3)--(0.55,4.3)--(0.55,3.9);
\end{tikzpicture}}
\caption{Two-user MIMO Gaussian MAC with one-sided feedback.\label{fig:mac_one}}
\end{figure}
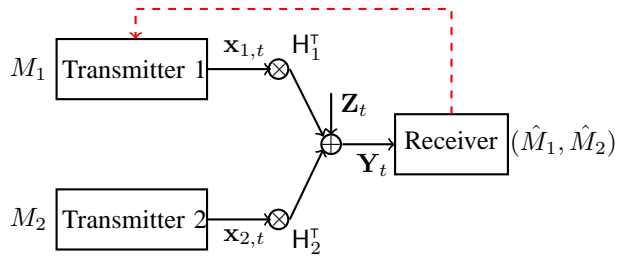

Consider the Gaussian MIMO MAC described in~\eqref{eq:mac_channel_output}, Section~\ref{sec:MAC}, but where only Transmitter~1 has feedback from the receiver (see Figure~\ref{fig:mac_one}). The inputs are thus of the form
\begin{subequations}
\begin{IEEEeqnarray}{rCL}\label{eq:MACinputsone}
\vect{X}_{1,t}&=&\varphi_{1,t}^{(n)} (M_1, \vect{Y}_1,\dots,\vect{Y}_{t-1})\\
\vect{X}_{2,t}&=&\varphi_{2,t}^{(n)} (M_2).
\end{IEEEeqnarray}
\end{subequations}

We will again restrict to linear-feedback schemes where the inputs at Transmitter~1 are generated as 
\begin{IEEEeqnarray}{rCl}
\vect{X}_{1,t}&=& \vect{W}_{1,t} + \sum_{\tau=1}^{t-1} \mat{C}_{1,\tau,t} \vect{Y}_\tau,\IEEEeqnarraynumspace\label{eq:lin_comone}
\end{IEEEeqnarray}
where $\vect{W}_{1,t}$ is a vector that only depends on the message $M_1$ but not on the feedback, $\vect{ W}_{1,t}= \xi_{1,t}^{(n)}(M_1)$ for arbitrary functions~$\xi_{1,t}^{(n)}$. 

Decoding, power constraint,  and the definitions of error probabilities and capacity regions are as described in Section~\ref{sec:MAC}. 

We denote the linear-feedback capacity region of the Gaussian MIMO MAC with one-sided feedback from Receiver~1 by $\set{C}_{\MAC, \One}^{\linfb}\left(\trans{\mat{H}}_1, \trans{\mat{H}}_2; P\right)$. It is unknown to date. Inner bounds (i.e., achievable regions) were presented in \cite{WVS83,carleial, CoverLeung, LapidothWigger}.

Analogous to Proposition~\ref{prop:mac},  we can derive a multi-letter expression for the linear-feedback capacity region 
$\set{C}_{\MAC, \One}^{\linfb}\left(\trans{\mat{H}}_1,\trans{\mat{H}}_2;P\right)$. Recall the definition of the regions 
$\set{R}_{\MAC}$ in Definition~\ref{def:RMAC}.
\begin{proposition}\label{prop:macone}
\begin{IEEEeqnarray*}{rCl}
\set{C}_{\MAC,\One}^{\linfb}\left(\trans{\mat{ H}}_1,\trans{\mat{ H}}_2;P\right)
&=&\textnormal{cl}\left(\bigcup_{\eta,\mat{D}_1^{\B}}\frac{1}{\eta}\set{R}_{\MAC}\left(\eta,\mat{D}_1^{\B},\mat{0},\tBHone,\tBHtwo;P\right)\right)
\IEEEeqnarraynumspace
\end{IEEEeqnarray*}
where the union is over all positive integers $\eta$ and all strictly-lower block-triangular $(\eta \nu_1)$-by-$(\eta \kappa)$ matrices $\mat{D}_1^{\B}$ with block sizes $\nu_1\times \kappa$ that satisfy 
$\textnormal{tr}\left(\mat{D}_1^{\B} \trans{\mat{(\mat{D}_1^{\B})}}\right)\leq\eta P,$ and where $\mat{0}$ denotes the $(\eta \nu_2)$-by-$(\eta \kappa)$ all-zero matrix.
\end{proposition}
\begin{proof}
Analogous to the proof of Proposition~\ref{prop:mac}, but where the matrix $\mat{D}_2^{\B}$ needs to be the $(\eta \nu_2)$-by-$(\eta \kappa)$ all-zero matrix which implies that $\{\mat{C}_{2, \tau,\ell} \}$ are all equal to the $\nu_2$-by-$\kappa$ all-zero matrix. \end{proof}

\subsection{Duality Result}
\begin{theorem}
\label{thm:dualone}
 \begin{equation}  
  \set{C}_{\BC,\One}^{\linfb}\left(\mat{H}_1,\mat{H}_2;P\right)
= \set{C}_{\MAC,\One}^{\linfb}\left(\trans{\mat{H}_1},\trans{\mat{H}_2};P\right).\end{equation}
\end{theorem}
\begin{proof}
Follows from Propositions~\ref{prop:bcone} and \ref{prop:macone} and Remark~\ref{rem:trans} which continues to hold in the one-sided feedback setup, and because  $\bar{\mat{0}}= \mat{0}$ and Propositon~\ref{prop:equal_regions} imply the following:

If $\mat{B}_1^{\B}=\bar{\mat{D}}_1^{\B}$, then \begin{IEEEeqnarray}{rCl}
\set{R}_\BC\left(\eta,\mat{B}_1^{\B},\mat{0},\mat{ H}_1^\B,\mat{ H}_2^\B;P\right)&=& 
\set{R}_\MAC\left(\eta,\mat{D}_1^{\B},\mat{0},\bBHone, \bBHtwo;P\right). \IEEEeqnarraynumspace
\end{IEEEeqnarray}
\end{proof}

\section{Extension II: $K\geq 2$ Users}
\label{sec:ext_many}
In this section we consider the $K$-user Gaussian BC and MAC with feedback, when $K\geq2$.
\subsection{$K\geq 2$-user MIMO Gaussian  BC with Feedback}
\begin{figure}[ht]
\centerline{
\begin{tikzpicture}[xscale=1,yscale=1]
\draw [thick] (0,2.1) rectangle (1.75,2.9);
\node at (5.25,2.55){$\vdots$};
\node[right] at (-0.1,2.55){Transmitter};
\node[left] at (0.08,2.55){$(M_1,\dots,M_K)$};
\draw [->][thick](1.75,2.5)--(2.4,3.05);
\draw [thick](2.5,3.1) circle [radius=0.13];
\node[above left] at (2.5,3.1){$\mat{H}_1$};
\node  at (2.5,3.1){$\times$};
\draw [->][thick](2.6,3.2)--(3,3.5)--(3.5,3.5);
\node[right] at (1.85,2.5){$\vect{x}_t$};
\draw [thick](3.63,3.5) circle [radius=0.13];
\node at (3.63,3.5){$+$};
\draw [->][thick](3.63,4)--(3.63,3.65);
\node [left] at (3.68,3.8){$\vect{Z}_{1,t}$};
\draw [->][thick](3.78,3.5)--(4.45,3.5);
\draw [thick] (4.45,3.1) rectangle (6.2,3.9);
\node[right] at (4.4,3.5){Receiver~1};
\node[right] at (6.1,3.5){$\hat{M}_1$};
\draw [->][thick](1.75,2.5)--(2.4,2);
\draw [thick](2.5,1.9) circle [radius=0.13];
\node[below left] at (2.5,1.8){$\mat{H}_K$};
\node  at (2.5,1.9){$\times$};
\draw [->][thick](2.6,1.83)--(3,1.5)--(3.5,1.5);
\draw [thick](3.63,1.5) circle [radius=0.13];
\node at (3.63,1.5){$+$};
\draw [->][thick](3.63,1)--(3.63,1.35);
\node [left] at (3.68,1.2){$\vect{Z}_{K,t}$};
\draw [->][thick](3.78,1.5)--(4.45,1.5);
\draw [thick] (4.45,1.1) rectangle (6.2,1.9);
\node[right] at (4.4,1.5){Receiver~$K$};
\node[right] at (6.1,1.5){$\hat{M}_K$};
\node[below] at (4.1,3.5){$\vect{Y}_{1,t}$};
\node[above] at (4.1,1.5){$\vect{Y}_{K,t}$};
\draw [->][dashed,thick][red](5.25,3.9)--(5.25,4.3)--(0.9,4.3)--(0.9,2.9);
\draw [->][dashed,thick][red](5.25,1.1)--(5.25,0.7)--(0.9,0.7)--(0.9,2.1); 
\end{tikzpicture}}
\caption{$K$-user MIMO Gaussian BC with feedback.\label{fig:bc_K}} 
\end{figure}
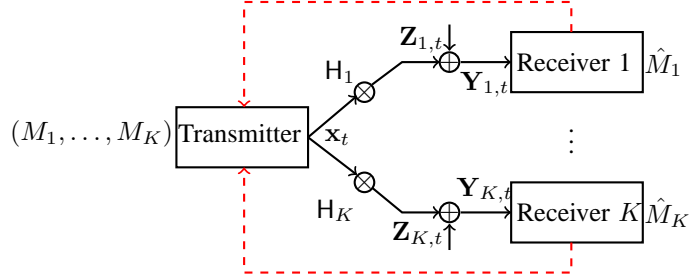
We consider the $K\geq 2$-receiver Gaussian BC with perfect output-feedback depicted in Figure~\ref{fig:bc_K}.
At each time $t\in \mathbb{N}$, if $\vect{x}_t$ denotes the  real vector-valued input symbol sent by the transmitter, 
Receiver~$i$, for $i\in\{1,\dots,K\},$ observes the real vector-valued channel output
\begin{IEEEeqnarray}{rCL}
\vect{Y}_{i,t}&=&\mat{H}_i \vect{x}_t+\vect{Z}_{i,t},
\end{IEEEeqnarray}
where $\mat{H}_i$ is a deterministic nonzero real $\nu_i$-by-$\kappa$ channel matrix known to transmitter and receivers, and the sequence of noises $\{(\vect{Z}_{1,t},\dots,\vect{Z}_{K,t})\}_{t=1}^n$ is a sequence of i.i.d.~centered Gaussian random vectors, each  of identity covariance matrix.

We will again restrict to linear-feedback schemes where the inputs, at each time $t\in \{1,\dots,n\}$, are generated as 
\begin{IEEEeqnarray}{rCl}
\vect{X}_t&=&\vect{W}_t+\sum_{i=1}^{K} \sum_{\tau=1}^{t-1} \mat{A}_{i,\tau,t}\vect{Y}_{i,\tau},  \label{eq:lin_comb_many_BC}
\end{IEEEeqnarray}
where $\vect{W}_t= \xi_t^{(n)}(M_1, \dots,M_K)$, for an arbitrary function~$\xi_t^{(n)}$, is thus a vector that only depends on the messages but not on the feedback.

Decodings, power constraint,  and the definitions of error probabilities and capacity regions are similar to Section~\ref{sec:BC} when we consider $K$ instead of two users. 

We denote the linear-feedback capacity region for this setup by $\set{C}_{\BC}^{\linfb}(\mat{H}_1, \dots,\mat{H}_K; P)$. It is unknown to date. Achievable regions are presented in \cite{OZAROW-LEUNG} and \cite{AMF12}.

Analogous to the definition of the regions $\set{R}_{\BC}$ in Definition~\ref{def:RBC}, we define  $\set{R}_{\BC}\left(\eta,\mat{B}_1^{\B},\dots,\mat{B}_K^{\B},\BHone,\dots,\BH_K;P\right)$ as the capacity region of the MIMO BC 
\begin{IEEEeqnarray}{rCl}
\label{eq:BCoutputsmany}
\vect{Y}_i&=&\BH_i \vect{U} +\BH_i \left(\sum_{j=1}^K \mat{B}_j^{\B} \vect{Z}_j\right)
+\vect{Z}_i,\quad i\in \{1,\dots,K\},\IEEEeqnarraynumspace
\end{IEEEeqnarray}
when the channel inputs $\vect{U}$ is average block-power constrained to 
\begin{equation}
\eta P -\sum_{j=1}^K \tr{\mat{B}_j^{\B} \trans{(\mat{B}_j^{\B})})}.
\end{equation}
\begin{proposition}\label{prop:bcmany}
\begin{IEEEeqnarray*}{rCl}
\set{C}_{\BC}^{\linfb}(\mat{H}_1,\ldots,\mat{ H}_K;P)&=&\textnormal{cl}\left(\bigcup_{\eta,\mat{B}_1^{\B},\dots,\mat{B}_K^{\B}}\frac{1}{\eta}\set{R}_{\BC}\left(\eta,\mat{B}_1^{\B},\dots,\mat{B}_K^{\B},\BHone,\dots,\BH_K;P\right)\right)\IEEEeqnarraynumspace
\end{IEEEeqnarray*}
where the union is over all positive integers $\eta$ and all strictly-lower block-triangular $(\eta \kappa)$-by-$(\eta \nu_i)$ matrices $\mat{B}_i^{\B}$ with blocks of sizes $\kappa\times \eta_i$, for $i\in \{1,\dots,K\}$, that satisfy $\sum_{j=1}^K \tr{\mat{B}_j^{\B} \trans{(\mat{B}_j^{\B})})} \leq \eta P.$
\end{proposition}
\begin{proof}
Similar to the proof of Proposition~\ref{prop:bc} if  the linear-feedback schemes described in Section~\ref{sec:schemeBC} and the converse are modified so as to allow for an arbitrary number $K\geq 2$ of users.
Details omitted.
\end{proof}

\subsection{$K\geq 2$-user MIMO Gaussian  MAC with Feedback}
\begin{figure}[ht]
\centerline{
\begin{tikzpicture}[xscale=1,yscale=1]
\draw [thick] (-0.6,3.1) rectangle (1.6,3.9);
\node[right] at (-0.65,3.5){Transmitter~1};
\node at (0.5,2.5){$\vdots$};
\node[left] at (-0.55,3.5){$M_1$};
\draw [->][thick](1.6,3.5)--(2.35,3.5);
\draw [->][thick](2.6,3.5)--(3.05,2.6);
\draw [thick](2.45,3.5) circle [radius=0.13];
\node[above right] at (2.5,3.5){$\trans{\mat{H}_1}$};
\node [right] at (2.2,3.5){$\times$};
\node[above] at (2,3.5){$\vect{x}_{1,t}$};
\draw [->][thick](1.6,1.5)--(2.35,1.5);
\draw [->][thick](2.6,1.5)--(3.05,2.44);
\draw [thick](2.45,1.5) circle [radius=0.13];
\node [right] at (2.2,1.5){$\times$};
\node[below right] at (2.5,1.5){$\trans{\mat{H}_K}$};
\node[below] at (2,1.5){$\vect{x}_{K,t}$};
\draw [thick] (-0.6,1.1) rectangle (1.6,1.9);
\node[right] at (-0.65,1.5){Transmitter~$K$};
\node[left] at (-0.55,1.5){$M_K$};
\draw [thick] (4,2.1) rectangle (5.5,2.9);
\node[right] at (4,2.55){Receiver};
\node[right] at (5.4,2.55){$(\hat{M}_1,\dots,\hat{M}_K)$};
\draw [->][thick](3.3,2.5)--(4,2.5);
\node[below] at (3.7,2.5){$\vect{Y}_t$};
\draw [thick](3.15,2.5) circle [radius=0.13];
\node at (3.15,2.5){$+$};
\draw [->][thick](3.15,3.18)--(3.15,2.62);
\node [right] at (3.15,3.02){$\vect{Z}_t$};
\draw [->][dashed,thick][red](4.75,2.9)--(4.75,4.3)--(0.55,4.3)--(0.55,3.9);
\draw [->][dashed,thick][red](4.75,2.1)--(4.75,0.7)--(0.55,0.7)--(0.55,1.1);
\end{tikzpicture}}
\caption{$K$-user MIMO Gaussian MAC with feedback.\label{fig:mac_K}}
\end{figure}
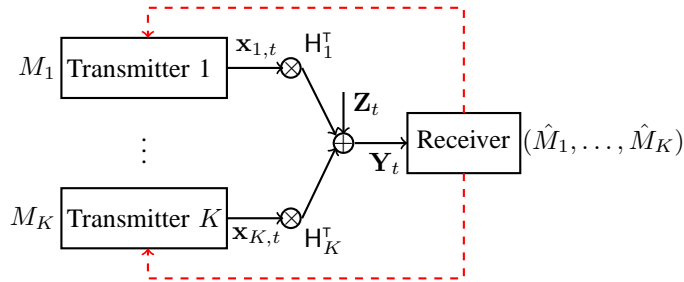

We consider the $K\geq 2$-transmitter Gaussian MAC with perfect output-feedback depicted in Figure~\ref{fig:mac_K}.
At each time $t\in \mathbb{N}$, if $\vect{x}_{i,t}$, for $i\in\{1,\dots,K\}$ denotes the  real vector-valued input symbol sent by Transmitter~$i$, the  
receiver observes the real vector-valued channel output
\begin{IEEEeqnarray}{rCL}
\vect{Y}_{t}&=&\sum_{i=1}^K\trans{\mat{H}_i} \vect{x}_{i,t}+\vect{Z}_{t},
\end{IEEEeqnarray}
where $\mat{H}_i$, for $i\in\{1,\dots,K\}$, is a constant nonzero real $\nu_i$-by-$\kappa$ channel matrix and the sequence of noises $\{\vect{Z}_{t}\}_{t=1}^n$  is a sequence of i.i.d.~centered Gaussian random vectors of identity covariance matrices.

We will again restrict to linear-feedback schemes where the inputs at Transmitter~$i$, for $i\in\{1,\dots,K\}$, are generated as 
\begin{IEEEeqnarray}{rCl}
\vect{X}_{i,t}&=& \vect{W}_{i,t} + \sum_{i=1}^K\sum_{\tau=1}^{t-1} \mat{C}_{i,\tau,t} \vect{Y}_\tau,\IEEEeqnarraynumspace
\end{IEEEeqnarray}
where $\vect{W}_{i,t}= \xi_{i,t}^{(n)}(M_i)$ 
for an arbitrary function~$\xi_{i,t}^{(n)}$ is thus a vector that only depends on the message $M_i$ but not on the feedback. 

Decoding, power constraint,  and the definitions of error probabilities and capacity regions are as described in Section~\ref{sec:MAC} extended to $K\geq2$ users. 
The linear-feedback capacity region is denoted by $\set{C}_{\MAC}^{\linfb}\left(\trans{\mat{H}}_1, \dots,\trans{\mat{H}}_K; P\right)$. It is unknown when $K>2$.  

We will be specially interested in the SISO case ($\nu_i=\kappa=1$) when the channel matrices $\mat{H}_1, \ldots, \mat{H}_K$ reduce to scalars $h_1,\ldots, h_K$. We denote the linear-feedback capacity region for this case by $\C_{\MAC,\SISO,\Sigma}^{\linfb}(h_1,\ldots, h_K;P)$.  Also this SISO capacity region is unknown when $K>2$. However, for equal channel coefficients $h_1=\ldots, h_K=h$, the results by Kramer~\cite{Kramer02} and Ardenistazadeh et al. \cite{AWKJ} combined with a symmetry argument as presented in Appendix~\ref{app:cor} immediately yield: 
\begin{equation}\label{eq:sumK}
\C_{\MAC,\SISO,\Sigma}^{\linfb}(h,\ldots, h;P) =\frac12 \log\left(1+P \phi(K,P)\right), 
\end{equation}
where $\phi(K,P)$ is the unique solution in $[1,K]$ to the following equation in $\phi$:
\begin{equation}
\left(1+P\phi\right)^{K-1}= \left(1+\frac{P}{K}\phi (K-\phi)\right).
\end{equation}

Analogous to the definition of the regions $\set{R}_{\MAC}$ in Definition~\ref{def:RMAC}, we define $\set{R}_{\MAC}\left(\eta,\mat{D}_1^{\B},\dots,\mat{D}_K^{\B},\tBHone,\dots,\trans{\mat{(\mat{H}_K^{\B})}};P\right)$ as the capacity region of the MIMO MAC 
\begin{IEEEeqnarray}{lCl}
\vect{Y}&=&\left(\mat{I}+\sum_{i=1}^K\trans{(\mat{H}_i^{\B})}\mat{D}_i^{\B}\right)\cdot \left(
\sum_{i=1}^K\trans{(\mat{H}_i^{\B})} \mat{Q}_i^{-1} \vect{U}_i +\vect{Z}\right),\label{eq:MACoutputmany}\IEEEeqnarraynumspace
\end{IEEEeqnarray}
when the inputs $\vect{U}_1,\dots,\vect{U}_K$ are average block-sumpower constrained to
\begin{equation}
\eta P- \sum_{j=1}^K \tr{\mat{D}_j^{\B} \trans{(\mat{D}_j^{\B})}},
\end{equation}
where $\mat{Q}_i$, for $i\in \{1,\dots, K\},$ is the unique positive square root of 
\begin{IEEEeqnarray}{lCl}
\mat{M}_i&=&\trans{(\mat{I}+\mat{D}_i^{\B}\trans{(\mat{H}_i^{\B})})}(\mat{I}+ \mat{D}_i^{\B}\trans{(\mat{H}_i^{\B})})+\sum_{j=1;j\neq i}^K\trans{(\mat{D}_j^{\B}\trans{(\mat{H}_j^{\B})})}(\mat{D}_j^{\B} \trans{(\mat{H}_j^{\B})}).\IEEEeqnarraynumspace
\end{IEEEeqnarray}

\begin{proposition}\label{prop:macmany}
\begin{IEEEeqnarray*}{l}
\set{C}_{\MAC}^{\linfb}(\trans{\mat{ H}}_1,\dots,\trans{\mat{ H}}_K;P)=
\textnormal{cl}\left(\bigcup_{\eta,\mat{D}_1^{\B},\dots,\mat{D}_K^{\B}}\frac{1}{\eta}\set{R}_{\MAC}\left(\eta,\mat{D}_1^{\B},\dots,\mat{D}_K^{\B},\tBHone,\dots,\trans{\mat{(\mat{H}_K^{\B})}};P\right)\right),
\end{IEEEeqnarray*}
where the union is over all positive integers $\eta$ and all strictly-lower block-triangular $(\eta \nu_i)$-by-$(\eta \kappa)$ matrices $\mat{D}_i^{\B}$ of blocks with sizes $\nu_i\times \kappa$, for $i\in \{1,\dots,K\}$, that satisfy
$\sum_{j=1}^K \tr{\mat{D}_j^{\B} \trans{(\mat{D}_j^{\B})}}\leq \eta P$.
\end{proposition}
\begin{proof}
Analogous to the proof of Proposition~\ref{prop:mac}, but where the linear-feedback schemes described in Section~\ref{sec:schemeMAC} and the converse need to be modified so as to allow for an arbitrary number $K\geq 2$ of users. 
Details omitted.
\end{proof}

\subsection{Duality Result}
Our main result on duality can  also be extended to the MIMO BC and MAC with more than two users.
\begin{theorem}
\label{thm:dual_many}
 The linear-feedback capacity regions of the $K\geq 2$-user MIMO Gaussian BC with channel matrices $\mat{H}_1,\dots,\mat{H}_K$ under sum-power constraint $P$ and the $K\geq 2$-user MIMO Gaussian MAC with channel matrices $\trans{\mat{H}_1},\dots,\trans{\mat{H}_K}$ under sum-power constraint $P$ coincide:
 \begin{equation}  
  \set{C}_{\BC}^{\linfb}\left(\mat{H}_1,\dots,\mat{H}_K;P\right)
= \set{C}_{\MAC}^{\linfb}\left(\trans{\mat{H}_1},\dots,\trans{\mat{H}_K};P\right).
\end{equation}
\end{theorem}
\proof{The proof follows by Proposition~\ref{prop:bcmany} and~\ref{prop:macmany}, Remark~\ref{rem:trans} which continues to hold for this setup, and Proposition~\ref{prop:equal_regions} which can be extended to $K\geq2$ users since the nofeedback MAC-BC duality holds for $K\geq2$ users~\cite{VJG03}.}

Specializing this theorem to the SISO case under equal channel gains $h_1=\ldots, h_K=h$, we obtain: 
\begin{corollary}
\begin{equation}\label{eq:maxsymsumrateK}
 C_{\BC,\SISO,\Sigma}^{\linfb}(h,\ldots, h;P) = C_{\MAC,\SISO,\Sigma}^{\linfb}(h,\ldots, h;P) 
 \end{equation}
where a computable expression for $C_{\MAC,\SISO,\Sigma}^{\linfb}(h,\ldots, h;P)$ is given in~\eqref{eq:sumK}.
\end{corollary} 
The achievability of the sum-rate in~\eqref{eq:maxsymsumrateK} for the $K$-user scalar Gaussian BC with equal channel gains was already established by the control-theory-inspired scheme in \cite{AMF12}. Our result here establishes that for the symmetric scalar Gaussian BC  and arbitrary number of users $K>2$ this scheme is indeed  sum-rate optimal among all linear-feedback schemes.

\section{Proofs}
\label{sec:proofs}

\subsection{Converse to Proposition \ref{prop:bc}}
\label{sec:con_prop_bc}
We wish to prove
\begin{IEEEeqnarray}{rCl}
\label{eq:conv_bc}
\set{C}_{\BC}^{\linfb}\left(\mat{H}_1,\mat{H}_2;P\right)
 &\subseteq& \textnormal{cl}\left( \bigcup_{(\eta,\mat{B}_1^\B, \mat{B}_2^\B)}   \frac{1}{\eta}  \mathcal{R}_{\BC}\left(\eta,\mat{B}_1^\B,\mat{ B}_2^\B,\mat{H}_1^\B,\mat{H}_2^\B;P\right) \right). \IEEEeqnarraynumspace
\end{IEEEeqnarray}
Fix $(R_1,R_2)\in\set{C}_{\BC}^\linfb\left(\mat{H}_1,\mat{H}_2;P\right)$ and for these rates and for each blocklength $n$ 
we fix encoding and decoding functions $\tilde{ \xi}^{(n)}, \phi_1^{(n)}, \phi_2^{(n)}$ and linear-feedback matrices $\{\mat{B}_{i,\tau,\ell}^{(n)}\}$ 
such that the sequence of probabilities of error $P_{\e,\BC}^{(n)}\to 0$ as $n\to\infty$ and the power constraint~\eqref{eq:powerconstraintBC} is satisfied for each $n$. (Thus, we use the form in \eqref{eq:lin_com_2_bc} to describe the channel inputs.)

Applying Fano's inequality, we obtain that for each $i\in\{1,2\}$ and for each positive integer $n$,  
\begin{IEEEeqnarray}{rCl}
n R_i
&\leq& I(M_i;\vect{Y}_i^{(n)})+\epsilon_n,
\end{IEEEeqnarray}
where $\frac{\epsilon_n}{n}\to 0$ as $n \to \infty$ and where $\vect{Y}_i^{(n)}$ denotes the $n\nu_i$-dimensional column-vector that is obtained by stacking on top of each other all the $n$ vectors observed at Receiver~$i$ when the blocklength-$n$ scheme is applied. 

Letting $n \to \infty$, we have
 \begin{IEEEeqnarray}{rCl}
 R_i&\leq&\varlimsup_{n\to \infty}  \frac1n I(M_i;\vect{Y}_i^{(n)}),\quad i \in \{1,2\}.\label{eq:rconv1}
 \end{IEEEeqnarray}
 
Since the RHS of \eqref{eq:conv_bc} is closed, it suffices to prove that for all $\delta >0,$ the pair $(R'_1,R'_2),$ 
\begin{subequations}\label{eq:Rdelta}
\begin{IEEEeqnarray}{rCl}
R'_1&\triangleq& \eta (R_1-\delta),\\
R'_2&\triangleq& \eta (R_2-\delta),
\end{IEEEeqnarray}
\end{subequations}
lies in $\mathcal{R}_{\BC}\left(\eta,\mat{B}_1^\B,\mat{B}_2^\B,\mat{H}_1^\B,\mat{H}_2^\B;P\right)$ for some positive integer~$\eta$ and strictly-lower block-triangular $\eta\kappa$-by-$\eta\nu_1$ and $\eta\kappa$-by-$\eta\nu_2$ matrices $\mat{B}_1^\B$ and $\mat{B}_2^\B$ of block sizes $\kappa\times\nu_1$ and $\kappa\times\nu_2$. 

By~\eqref{eq:rconv1} and \eqref{eq:Rdelta}, there exists a finite blocklength $n$ such that
\begin{subequations}
\label{eq:BC_r_prim}
\begin{IEEEeqnarray}{rCl}
R'_1&\leq&I(M_1;\vect{Y}_1^{(n)}), \\ 
R'_2&\leq& I(M_2;\vect{Y}_2^{(n)}).
\end{IEEEeqnarray}
\end{subequations}
In the sequel, let $n$ be so that \eqref{eq:BC_r_prim} holds. Also, based on the parameters $\{\mat{B}_{i,\tau,\ell}^{(n)}\}$ of the blocklength-$n$ scheme, define
\begin{IEEEeqnarray}{rCl}
\mat{B}_i^{\B}&=&\begin{bmatrix}   
\mat{0}&& \dots &&\mat{0} \\ 
\mat{B}_{i,1,2}^{(n)} &\mat{0}  \\
\mat{B}_{i,1,3}^{(n)}&\mat{B}_{i,2,3}^{(n)} &\mat{0} \\ 
\vdots&&&\ddots\\
\mat{B}_{i,1,n}^{(n)}&\mat{B}_{i,2,n}^{(n)}&\dots&\mat{B}_{i,(n-1),n}^{(n)}&\mat{0}
\end{bmatrix},\quad i \in \{1,2\}.
\end{IEEEeqnarray}
The corresponding channel outputs $\vect{Y}_1^{(n)}$ and $\vect{Y}_2^{(n)}$ are
\begin{subequations}
\label{eq:BC_outputs}
\begin{IEEEeqnarray}{rCl}
\vect{Y}_1^{(n)}& = \BHone \vect{\tilde{W}}^{(n)} + (\mat{I}+\BHone\mat{B}_1^{\B}) \vect{Z}_1^{(n)} + \BHone \mat{B}_2^{\B} \vect{Z}_2^{(n)},\IEEEeqnarraynumspace\\
\vect{Y}_2^{(n)}& = \BHtwo \vect{\tilde{W}}^{(n)} + (\mat{I}+\BHtwo\mat{B}_2^{\B} )\vect{Z}_2^{(n)} + \BHtwo \mat{B}_1^{\B} \vect{Z}_1^{(n)},\IEEEeqnarraynumspace
\end{IEEEeqnarray}
\end{subequations}
where $\vect{Z}_1^{(n)}= \trans{\begin{pmatrix}\trans{\vect{Z}_{1,1}},\dots,\trans{\vect{Z}_{1,n}}\end{pmatrix}}$, $\vect{Z}_2^{(n)}=\trans{\begin{pmatrix}\trans{\vect{Z}_{2,1}},\dots,\trans{\vect{Z}_{2,n}}
\end{pmatrix}}$, and  $\vect{ \tilde W}^{(n)}$ is the $n\kappa$-dimensional vector that is obtained when stacking on top of each other all the $n$ codevectors ($\kappa$-dimensional column-vectors) that are produced by the encoding function $\tilde{\xi}^{(n)}$. Notice that the power-constraint~\eqref{eq:powerconstraintBC} is  equivalent to requiring that 
\begin{equation}
\E{\|\vect{\tilde W}^{(n)}\|^2}\leq n P - \tr{\mat{B}_1^\B\trans{\mat{(\mat{B}_1^\B)}}}- \tr{\mat{B}_2^\B\trans{\mat{(\mat{B}_2^\B})}}.
\end{equation}

Let now $\eta=n$ and consider the BC in~\eqref{eq:BC_outputs} where the transmitter is equipped with $\eta \kappa$ antennas and Receiver~$i$ with $\eta \nu_i$ antennas, for $i\in\{1,2\}$, and where $\vect{\tilde W}^{(\eta)}$ denotes the $\eta\kappa$-dimensional input-vector. Recall that we denoted by $\mathcal{R}_{\BC}(\eta,\mat{B}_1^\B,\mat{B}_2^\B,\mat{H}_1^\B,\mat{H}_2^\B;P)$  the capacity region of this channel under an expected average block-power constrained $(\eta P - \tr{\mat{B}_1^\B\trans{\mat{(\mat{B}_1^\B)}}}- \tr{\mat{B}_2^\B\trans{\mat{(\mat{B}_2^\B})}})$ on the input $\vect{\tilde W}^{(\eta)}$.
Using random coding and joint typicality decoding, it can be shown that the nonnegative rate-pair $(\tilde{R}_1,\tilde{R}_2)$ lies in this capacity region $\mathcal{R}_{\BC}(\eta,\mat{B}_1^\B,\mat{B}_2^\B,\mat{H}_1^\B,\mat{H}_2^\B;P)$ if it satisfies 
\begin{subequations}
\begin{IEEEeqnarray}{cCl}
\tilde{R}_1&\leq& I(\Theta_1;\vect{Y}_1^{(\eta)})\\
 \tilde{R}_2&\leq& I(\Theta_2;\vect{Y}_2^{(\eta)})
\end{IEEEeqnarray}
\end{subequations}
for some independent auxiliary random variables $\Theta_1$ and $\Theta_2$ and a choice of  $\vect{\tilde W}^{(\eta)}$ such that $(\Theta_1,\Theta_2, \vect{\tilde W}^{(\eta)})$ are independent  of $(\vect{Z}_1^{(\eta)},\vect{Z}_2^{(\eta)})$. 

Specializing this last argument to $\Theta_1=M_1$ and $\Theta_2=M_2$, by \eqref{eq:BC_r_prim}, we conclude that for any $\delta>0$ the rate-pair $(R_1', R_2')$ defined in \eqref{eq:Rdelta} lies in $\mathcal{R}_{\BC}(\eta,\mat{B}_1^\B,\mat{B}_2^\B,\mat{H}_1^\B,\mat{H}_2^\B;P)$, which concludes the proof.

\subsection{Proof of Lemma~\ref{lem1}}\label{sec:pflm2}
For the inputs transmitted in the first $\eta$-length block and described by~\eqref{eq:input_mac}, it holds:
\begin{IEEEeqnarray}{rCl}
\E{\|\vect{X}_{1}\|^2} & = &  \tr{ \E{\vect{X}_1 \trans{\vect{X}_1}}} \nonumber \\
 & = & \tr{(\mat{Q}_1^{-1} +\mat{D}_1^\B\tBHone \mat{Q}_1^{-1} )  \E{\vect{U}_1 \trans{\vect{U}}_1}  \trans{(\mat{Q}_1^{-1}  +\mat{D}_1^\B\tBHone \mat{Q}_1^{-1} )}}\nonumber \\
 &  &  +  \tr{\mat{D}_1^\B  \tBHtwo \mat{Q}_2^{-1} \E{\vect{U}_2 \trans{\vect{U}}_2} \trans{\mat{D}_1^\B (\tBHtwo \mat{Q}_2^{-1})} }+ \tr{\mat{D}_1^\B \trans{\mat{(\mat{D}_1^\B)}} },\nonumber
\end{IEEEeqnarray}
and similarly, 
\begin{IEEEeqnarray}{rCl}
\E{\|\vect{X}_{2}\|^2} & = &  \tr{ \E{\vect{X}_2 \trans{\vect{X}_2}}} \nonumber \\
 & = & \tr{(\mat{Q}_2^{-1}+\mat{D}_2^\B\tBHtwo\mat{Q}_2^{-1} )  \E{\vect{U}_2 \trans{\vect{U}}_2}  \trans{(\mat{Q}_2^{-1}+\mat{D}_2^\B\tBHtwo \mat{Q}_2^{-1} )}  }\nonumber \\
 &  &     +\tr{\mat{D}_2^\B  \tBHone \mat{Q}_1^{-1} \E{\vect{U}_1 \trans{\vect{U}}_1}(\mat{D}_2^\B\tBHone \mat{Q}_1^{-1})}+ \tr{ \mat{D}_2^\B \trans{\mat{(\mat{D}_2^\B)}}}.\nonumber
\end{IEEEeqnarray}
Notice now that since $\tr{\mat{AB}}=\tr{\mat{BA}}$, by the definition of $\mat{M}_1$ in~\eqref{eq:M1} and because $\mat{Q}_1^{-1}$ is symmetric
\begin{IEEEeqnarray}{rCl}\lefteqn{
\textnormal{tr}\Big( (\mat{Q}_1^{-1}+\mat{D}_1^\B\tBHone \mat{Q}_1^{-1} )  \E{\vect{U}_1 \trans{\vect{U}}_1}  \trans{(\mat{Q}_1^{-1}+\mat{D}_1^\B\tBHone \mat{Q}_1^{-1} )}\Big)}\qquad\nonumber \\ 
&& \hspace{1cm} + \tr{\mat{D}_2^\B  \tBHone \mat{Q}_1^{-1} \E{\vect{U}_1 \trans{\vect{U}}_1} \trans{(\mat{D}_2^\B\tBHone \mat{Q}_1^{-1})} }\nonumber\\
& = & \textnormal{tr}\Big( \mat{Q}_1^{-1}  \big( \trans{(\mat{I}+\mat{D}_1^\B\tBHone)}  (\mat{I}+\mat{D}_1^\B\tBHone) \nonumber \\ 
& & \hspace{1cm}+ \trans{(\mat{D}_2^\B \tBHone)}( \mat{D}_2^\B  \tBHone) \big) \mat{Q}_1^{-1}  \cdot\E{\vect{U}_1 \trans{\vect{U}}_1} \Big)\nonumber \\
& = & \tr{  {\mat{Q}}_1^{-1} \mat{M}_1  \mat{Q}_1^{-1} \E{\vect{U}_1 \trans{\vect{U}}_1} }\nonumber \\
&= & \E{\|\vect{U}_1\|^2}. 
\end{IEEEeqnarray}
Similarly,
\begin{IEEEeqnarray}{rCl}\lefteqn{
\textnormal{tr}\big( (\mat{Q}_2^{-1}+\mat{D}_2^\B\tBHtwo \mat{Q}_2^{-1} )  \E{\vect{U}_2 \trans{\vect{U}}_2}  \trans{(\mat{Q}_2^{-1}+\mat{D}_2^\B\tBHtwo \mat{Q}_2^{-1} )} \big)}\quad\nonumber \\
& & \hspace{1cm} + \tr{\mat{D}_1^\B  \tBHtwo \mat{Q}_2^{-1} \E{\vect{U}_2 \trans{\vect{U}}_2} \trans{(\mat{D}_1^\B\tBHtwo\mat{Q}_2^{-1})} }\nonumber \\
&= & \E{\|\vect{U}_2\|^2}. \hspace{6.5cm}
\end{IEEEeqnarray}
Combining all these equalities, by the linearity of the trace, we obtain that 
\begin{IEEEeqnarray*}{rCl}\E{\|\vect{X}_{1}\|^2}+ \E{\|\vect{X}_{2}\|^2} 
 &=&\E{\|\vect{U}_1\|^2}+  \E{\|\vect{U}_2\|^2}
 + \tr{\mat{D}_1^\B \trans{\mat{(\mat{D}_1^\B)}}}+\tr{\mat{D}_2^\B \trans{\mat{(\mat{D}_2^\B)}}}
\end{IEEEeqnarray*}
and can thus conclude that the input sequences satisfy the average total input-power constraint $P$ whenever  $\eta P- \tr{\mat{D}_1^\B \trans{\mat{(\mat{D}_1^\B)}}}-\tr{\mat{D}_2^\B \trans{\mat{(\mat{D}_2^\B)}}}\geq 0$ and the vectors $\vect{U}_1$ and $\vect{U}_2$ produced by the outer code satisfy the average total input-power constraint $\left(\eta P- \tr{\mat{D}_1^\B \trans{\mat{(\mat{D}_1^\B)}}}-\tr{\mat{D}_2^\B \trans{\mat{(\mat{D}_2^\B)}}}\right)$.

\subsection{Converse Proof to Proposition~\ref{prop:mac}}
\label{sec:conv_prop_mac}
We wish to prove
 \begin{IEEEeqnarray}{rCl}
 \label{eq:conv_mac}
 \set{C}_{\MAC}^{\linfb}\left(\trans{\mat{H}_1},\trans{\mat{H}_2};P\right)
&\subseteq&  \textnormal{cl}\left( \bigcup_{(\eta,\mat{D}_1^\B, \mat{D}_2^\B)}\frac{1}{\eta}  \mathcal{R}_{\MAC}\left(\eta,\mat{D}_1^\B,\mat{D}_2^\B,\tBHone,\tBHtwo;P\right)\right)\hspace*{-1mm}.\IEEEeqnarraynumspace
 \end{IEEEeqnarray}

Fix $(R_1,R_2)\in\set{C}_{\MAC}^{\linfb}\left(\trans{\mat{H}_1},\trans{\mat{H}_2};P\right)$ and for these rates and for each blocklength $n$ 
we fix encoding and decoding functions $\xi_1^{(n)}, \xi_2^{(n)}, \phi^{(n)},$ and linear-feedback matrices $\{\mat{C}_{i,\tau,\ell}^{(n)}\}$ 
such that the sequence of probabilities of error $P_{\e,\MAC}^{(n)}\to 0$ as $n\to\infty$ and the power constraint~\eqref{eq:powerconstraintMAC} is satisfied.

Applying Fano's inequality, we obtain that for each positive integer $n$,  
\begin{subequations}\label{eq:RMACdelta}
\begin{IEEEeqnarray}{cCl}
n R_1&\leq &I(M_1;\vect{Y}^{(n)}) 
+\epsilon_n,\\
 n R_2&\leq& I(M_2;\vect{Y}^{(n)})  
 +\epsilon_n, 
 \end{IEEEeqnarray} 
 \end{subequations}
where $\frac{\epsilon_n}{n}\to 0$ as $n \to \infty$ and where $\vect{Y}^{(n)}$ denotes the $n\kappa$-dimensional column-vector that is obtained by stacking on top of each other all the $n$ vectors observed at the receiver when the blocklength-$n$ scheme is applied. 

Letting $n \to \infty$, we have
 \begin{subequations}
 \label{eq:rconv2}
 \begin{IEEEeqnarray}{cCl}
 R_1&\leq& \varlimsup_{n\to \infty} \frac1n  I(M_1;\vect{Y}^{(n)})\\        
 R_2&\leq& \varlimsup_{n\to \infty} \frac1n  I(M_2;\vect{Y}^{(n)}).
 \end{IEEEeqnarray}
 \end{subequations}

Since the RHS of \eqref{eq:conv_mac} is closed, it suffices to prove that $\forall \delta >0,$ the pair $(R'_1,R'_2),$ 
\begin{subequations}
\label{eq:bc_r_prim2}
\begin{IEEEeqnarray}{rCl}
R'_1&\triangleq& \eta (R_1-\delta)\\
R'_2&\triangleq& \eta (R_2-\delta),
\end{IEEEeqnarray}
\end{subequations}
lies in $ \mathcal{R}_{\MAC}(\eta,\mat{D}_1^\B,\mat{D}_2^\B,\tBHone,\tBHtwo;P)$ for some positive integer $\eta$ and strictly-lower block-triangular $\eta\nu_1$-by-$\eta \kappa$ and $\eta\nu_2$-by-$\eta \kappa$ matrices $\mat{D}_1^\B$ and $\mat{D}_2^\B$ of block sizes $\nu_1\times\kappa$ and $\nu_2\times\kappa$, respectively.

By~\eqref{eq:rconv2} and~\eqref{eq:bc_r_prim2}, there exists a finite blocklength $n$ such that
\begin{subequations}
\label{eq:MAC_r_prim}
 \begin{IEEEeqnarray}{cCl}
 R'_1 &\leq&  I(M_1;\vect{Y}^{(n)}), \\
 R'_2 &\leq&  I(M_2;\vect{Y}^{(n)}). 
 \end{IEEEeqnarray}
 \end{subequations}
In the sequel, let $n$ be fixed and so that \eqref{eq:MAC_r_prim} holds. Also, based on the parameters $\{\mat{C}_{i,\tau,\ell}^{(n)}\}$ of the blocklength-$n$ scheme, let
\begin{IEEEeqnarray}{rCl}
\mat{C}_i^{\B}&=&\begin{bmatrix}   
\mat{0}&& \dots &&\mat{0} \\ 
\mat{C}_{i,1,2}^{(n)} &\mat{0}  \\
\mat{C}_{i,1,3}^{(n)}&\mat{C}_{i,2,3}^{(n)} &\mat{0} \\ 
\vdots&&&\ddots\\
\mat{C}_{i,1,n}^{(n)}&\mat{C}_{i,2,n}^{(n)}&\dots&\mat{C}_{i,(n-1),n}^{(n)}&\mat{0}
\end{bmatrix},\quad i\in \{1,2\},
\end{IEEEeqnarray}
and 
\begin{equation}
\mat{D}_i^{\B}=\mat{C}_i^{\B}\left(\mat{I}-\tBHone \mat{C}^{\B}_1-\tBHtwo \mat{C}_2^{\B}\right)^{-1},\quad i\in \{1,2\}.
\end{equation}
Let moreover, 
 $\mat{Q}_1$ and $\mat{Q}_2$ be the unique positive square roots of the (positive-definite)  matrices
\begin{subequations}
\begin{IEEEeqnarray*}{rCl}
\mat{M}_1&=&\trans{(\mat{I}+\mat{D}_1^{\B}\tBHone)}(\mat{I}+ \mat{D}_1^{\B}\tBHone)+\trans{(\mat{D}_2^{\B}\tBHone)}\mat{D}_2^{\B} \tBHone\\
 \mat{M}_2&=&\trans{(\mat{I}+\mat{D}_2^{\B}\tBHtwo)} (\mat{I}+ \mat{D}_2^{\B}\tBHtwo)+\trans{(\mat{D}_1^{\B} \tBHtwo)}\mat{D}_1^{\B}\tBHtwo
 \end{IEEEeqnarray*}
\end{subequations}
and define
\begin{subequations}
\begin{IEEEeqnarray}{rCl} 
\vect{U}_1^{(n)}& \triangleq& \mat{Q}_1 \vect{W}_2^{(n)} \\
\vect{U}_2^{(n)}& \triangleq& \mat{Q}_2 \vect{W}_2^{(n)} 
\end{IEEEeqnarray}
\end{subequations}
where $\vect{W}_i^{(n)}$ denotes the $n\nu_i$-dimensional column-vector that is obtained by stacking on top of each other all the $n$ vectors produced by the encoding function $\xi_i^{(n)}$.  

Using similar algebraic manipulations as leading to \eqref{eq:blockoutputs}, we can write $\vect{Y}^{(n)}$ as
\begin{IEEEeqnarray}{lCl}
\label{eq:MAC_output}
\vect{Y}^{(n)}&=&(\mat{I}+\tBHone\mat{D}_1^{\B}+\tBHtwo\mat{D}_2^{\B})\cdot\big(\tBHone \mat{Q}_1^{-1} \vect{U}_1^{(n)} + \tBHtwo \mat{Q}_2^{-1} \vect{U}_2^{(n)}+ \vect{Z}^{(n)}\big), \IEEEeqnarraynumspace
\end{IEEEeqnarray}
where $\vect{Z}^{(n)}=\trans{\begin{pmatrix}\trans{\vect{Z}_1},&\dots,&\trans{\vect{Z}_n}\end{pmatrix}}$. 
In the same way as in Lemma~\ref{lem1} it can be shown  that the power constraint~\eqref{eq:powerconstraintMAC} is equivalent to requiring that 
\begin{IEEEeqnarray}{rCl}
\E{\|\vect{U}_1^{(n)}\|^2}+\E{\|\vect{U}_2^{(n)}\|^2}
&\leq&\eta P-  \textnormal{tr}(\mat{D}_1^{\B} \trans{\mat{(\mat{D}_1^{\B})}})- \textnormal{tr}(\mat{D}_2^{\B} \trans{\mat{(\mat{D}_2^{\B})}}). 
 \end{IEEEeqnarray}

Let now $\eta=n$ and consider the MIMO MAC~\eqref{eq:MAC_output}, where Transmitter $i$, for $i\in \{1,2\}$, is equipped with $\eta \nu_i$ antennas, the receiver is equipped with $\eta \kappa$ antennas,
and where $\vect{U}_1^{(\eta)}$ and $\vect{U}_2^{(\eta)}$  denote the $\eta \nu_1$ and $\eta \nu_2$-dimensional independent input-vectors. Recall that we denoted by $\mathcal{R}_{\MAC}(\eta,\mat{D}_1^\B,\mat{D}_2^\B,\mat{H}_1^\B,\mat{H}_2^\B;P)$  the capacity region of this channel under an expected total average block-power constraint $(\eta P - \tr{\mat{D}_1^\B\trans{\mat{(\mat{D}_1^\B)}}}- \tr{\mat{D}_2^\B\trans{\mat{(\mat{D}_2^\B})}})$ on the inputs $\vect{U}_1^{(\eta)}$ and $\vect{U}_2^{(\eta)}$.
Using random coding and joint typicality decoding, 
it can be shown that the nonnegative rate-pair  $(\tilde R_1,\tilde R_2)$ lies in $\mathcal{R}_{\MAC}\left(\eta,\mat{D}_1^\B,\mat{D}_2^\B,\tBHone,\tBHtwo;P\right)$ if it satisfies
\begin{subequations}
\begin{IEEEeqnarray}{cCl}                 
  \tilde R_1&\leq&    I(\Theta_1; \vect{Y}^{(\eta)}), \label{eq:cond1}\\        
  \tilde R_2&\leq &   I(\Theta_2;\vect{Y}^{(\eta)}) \label{eq:cond2}
  \end{IEEEeqnarray}
  \end{subequations}
for some  auxiliary random variables $\Theta_1$ and $\Theta_2$ and some choice of the inputs $\vect{U}_1^{(\eta)}$ and $\vect{U}_2^{(\eta)}$ such that the pairs $(\Theta_1, \vect{U}_1^{(\eta)})$ and $(\Theta_2, \vect{U}_2^{(\eta)})$ are independent of each other and of the noise vectors $\vect{Z}_1^{(\eta)}, \vect{Z}_2^{(\eta)}$.
   
  Specializing this last argument to $\Theta_1=M_1$ and $\Theta_2=M_2$, by \eqref{eq:MAC_r_prim}, we conclude that the rate-pair $(R_1', R_2')$ defined in \eqref{eq:bc_r_prim2} lies in $\mathcal{R}_{\MAC}\left(\eta,\mat{D}_1^\B,\mat{D}_2^\B,\tBHone,\tBHtwo;P\right)$, which establishes the desired proof.

\subsection{Proof of Proposition~\ref{prop:equal_regions}}\label{sec:proofprop3}
Fix $\eta$, channel matrices $\mat{H}_1$ and $\mat{H}_2$, and strictly-lower block-triangular matrices $\mat{B}_1^{\B}, \mat{B}_2^{\B}$. Also, let $\mat{D}_1^{\B}, \mat{D}_2^{\B}$ be given as in~\eqref{eq:choice}. Notice that since $\mat{B}_1^{\B}$ and $\mat{B}_2^{\B}$ are strictly-lower block-triangular, so are $\mat{D}_1^{\B}$ and $\mat{D}_2^{\B}$.
Also, let $\BHone$ and $\BHtwo$ be defined by~\eqref{eq:blockmat} and for $i\in\{1,2\}$ let $\mat{\bar{H}}_i^{\B} = \mat{I}_\eta \otimes\mat{\bar{H}}_i$. 

We consider the MIMO MAC in~\eqref{eq:blockoutputs}, but where now $\trans{\mat{H}}_i$ and $\trans{(\BH_i)}$ are replaced by $\mat{\bar H}_i$ and $\bBH_i$. So, we consider the MIMO MAC: 
\begin{IEEEeqnarray}{rCl}\label{eq:blockoutputs2}
\vect{Y}'&=&(\mat{I}+\bBHone\mat{D}_1^{\B}+\bBHtwo\mat{D}_2^{\B})\cdot (\bBHone \mat{Q}_1^{-1}\vect{U}_1 + \bBHtwo\mat{Q}_2^{-1} \vect{U}_2+ \vect{Z}),
\end{IEEEeqnarray}
where now $\mat{Q}_1$ and $\mat{Q}_2$ are the unique positive-definite square-roots of the matrices
\begin{subequations}
\begin{IEEEeqnarray}{rCl}
\mat{M}_1&=&\trans{(\mat{I}+\mat{D}_1^{\B}\bBHone)}(\mat{I}+ \mat{D}_1^{\B}\bBHone)+\trans{(\mat{D}_2^{\B}\bBHone)}(\mat{D}_2^{\B} \bBHone),\IEEEeqnarraynumspace\\
 \mat{M}_2&=&\trans{(\mat{I}+\mat{D}_2^{\B}\bBHtwo)} (\mat{I}+ \mat{D}_2^{\B}\bBHtwo)+\trans{(\mat{D}_1^{\B} \bBHtwo)}(\mat{D}_1^{\B}\bBHtwo).
 \end{IEEEeqnarray}
\end{subequations}
That means $\mat{Q}_1$ and $\mat{Q}_2$ are the unique positive-definite symmetric matrices that satisfy 
\begin{subequations}\label{eq:Q1}
\begin{IEEEeqnarray}{rCl}
\mat{Q}_1 {\mat{Q}}_1 &= & \mat{M}_1\\
\mat{Q}_2 {\mat{Q}}_2 &= & \mat{M}_2.
\end{IEEEeqnarray}
\end{subequations}

Since the matrix $(\mat{I}+\bBHone\mat{D}_1^{\B}+\bBHtwo\mat{D}_2^{\B})$ is invertible, the capacity region of the MAC in~\eqref{eq:blockoutputs2} under any input power constraint equals the capacity region of the MAC 
\begin{IEEEeqnarray}{lCl}\label{eq:MIMOMAC}
\vect{Y}'_\MAC&=&\bBHone \mat{Q}_1^{-1}\vect{U}_1 + \bBHtwo\mat{Q}_2^{-1} \vect{U}_2+ \vect{Z}\label{eq:y_mac2}\IEEEeqnarraynumspace
\end{IEEEeqnarray}
under the same input power constraint. This holds because the receiver can multiply its output vectors by an invertible matrix without changing the capacity region of the MAC. 

We now turn to the BC~\eqref{eq:outputsBC}. 
Let $\mat{S}_1$ and $\mat{S}_2$ be the positive square roots of the positive-definite matrices
\begin{subequations} 
\begin{IEEEeqnarray}{rCl}
\mat{N}_1&\triangleq & (\mat{I}+\BHone\mat{B}_1^{\B})\trans{(\mat{I}+\BHone\mat{B}_1^{\B})}+(\BHone\mat{B}_2^{\B})\trans{(\BHone\mat{B}_2)}\\ 
\mat{N}_2&\triangleq &(\BHtwo \mat{B}_1^{\B}) \trans{(\BHtwo\mat{B}_1^{\B})}+(\mat{I}+\BHtwo\mat{B}_2^{\B})\trans{(\mat{I}+\BHtwo\mat{B}_2^{\B})}.
\end{IEEEeqnarray}
\end{subequations} 
That means, $\mat{S}_1$ and $\mat{S}_2$ are the unique positive-definite symmetric matrices that satisfy
\begin{subequations}\label{eq:S1}
\begin{IEEEeqnarray}{rCl}
\mat{S}_1 {\mat{S}}_1 &= & \mat{N}_1\\
\mat{S}_2 {\mat{S}}_2 &= & \mat{N}_2.
\end{IEEEeqnarray}
\end{subequations}
The matrices $\mat{S}_1$ and $\mat{S}_2$ are invertible. 
Therefore, since in a MIMO BC each receiver can multiply its output vectors by an invertible matrix (here $\mat{E}\mat{S}_i^{-1}$) without changing the capacity of the BC, under any power constraint on the input vectors $\vect{W}$, the MIMO BC in~\eqref{eq:outputsBC} has the same capacity region as the MIMO BC
\begin{IEEEeqnarray}{rCl}
\vect{{Y}}_{i}'&\triangleq&\mat{E}\mat{S}_i^{-1}{\BH_i} \vect{U}+ \vect{\tilde Z}_i,\quad i\in\{1,2\},\label{eq:new_mimo_bc1}
\end{IEEEeqnarray}
where  $\vect{\tilde Z}_1$ and $\vect{\tilde Z}_2$ denote independent centered Gaussian vectors of identity covariance matrices. 

Define  now a new input-vector $\vect{\breve{U}}$ which is obtained from ${\vect{U}}$ by reversing the order of the elements:
\begin{equation}\label{eq:Wtilde}
\vect{\breve{U}}\triangleq \mat{E} \vect{U}.
\end{equation}
Notice that $\|\vect{\breve{U}}\|^2$ and $\|\vect{U}\|^2$ are equal. Thus, when the input vectors $\vect{U}$ are average block-power constrained to 
\begin{equation}\label{power3}
\eta P- \textnormal{tr}(\mat{B}_1^{\B} \trans{\mat{(\mat{B}_1^{\B})}})- \textnormal{tr}(\mat{B}_2^{\B} \trans{\mat{(\mat{B}_2^{\B})}}),
\end{equation}
the MIMO BC in~\eqref{eq:new_mimo_bc1} has the same capacity region as  the MIMO BC
\begin{IEEEeqnarray}{rCl}
\vect{{Y}}'_{i,\BC}&\triangleq&\mat{E}\mat{S}_i^{-1}\BH_i \mat{E} \vect{\breve{U}}+ \vect{\tilde Z}_i,\quad i\in\{1,2\},\label{eq:new_mimo_bc}
\end{IEEEeqnarray}
when the input vectors $\vect{\breve{U}}$ are average block-power constrained to the same power~\eqref{power3}. 

We conclude the proof by showing that the capacity region of the  MIMO BC in~\eqref{eq:new_mimo_bc} under average input power constraint~\eqref{power3} 
and the capacity region of the MIMO MAC~\eqref{eq:MIMOMAC} under average input-power constraint 
\begin{equation}\label{power4} 
\eta P- \textnormal{tr}(\mat{D}_1^{\B} \trans{\mat{(\mat{D}_1^{\B})}})- \textnormal{tr}(\mat{D}_2^{\B} \trans{\mat{(\mat{D}_2^{\B})}})
\end{equation}
are the same.
To this end, we first notice that by Assumption~\eqref{eq:choice}, the two power constraints~\eqref{power3} and \eqref{power4} coincide. In fact,  for $i\in\{1,2\}$, 
\begin{IEEEeqnarray}{rCl}
\textnormal{tr}(\mat{B}_i^{\B} \trans{(\mat{B}_i^{\B})}) & = & \textnormal{tr}( \mat{E}\trans{(\mat{D}_i^{\B})} \mat{E} {\mat{E}\mat{D}_i^{\B} \mat{E} })\nonumber \\
&=&\textnormal{tr}( \mat{E}\trans{(\mat{D}_i^{\B})} \mat{D}_i^{\B} \mat{E} )\nonumber \\
&=&\textnormal{tr}(  \mat{D}_i^{\B} \mat{E} \mat{E}\trans{(\mat{D}_i^{\B})})\nonumber \\
& = & \textnormal{tr}(\mat{D}_i^{\B}\trans{(\mat{D}_i^{\B})} ),
\end{IEEEeqnarray}
where  the first, second, and fourth equality hold because $\mat{E}=\trans{\mat{E}}$ and $\mat{E}^{-1}=\mat{E}$, and the third equality holds because $\tr{\mat{A}\mat{B}}=\tr{\mat{B}\mat{A}}$ for any matrices $\mat{A}$ and $\mat{B}$. 
Moreover, we shall shortly show that the BC in~\eqref{eq:new_mimo_bc} and the MAC in~\eqref{eq:MIMOMAC} are dual in the sense that 
\begin{equation}\label{eq:condduality}
\mat{E}\mat{S}_i^{-1} \BH_i\mat{E}= \trans{(\bBH_i\mat{Q}_i^{-1})}, \qquad i\in\{1,2\}.
\end{equation}
The desired equality~\eqref{eq:equal_regions} in the proposition follows then immediately from  the nofeedback duality of the MIMO Gaussian MAC and BC, $\set{C}_{\BC}^{\textnormal{nofb}}(\mat{H}_1^{\B}, \mat{H}_2^{\B};\eta P)=\set{C}_{\MAC}^{\textnormal{nofb}}(\trans{(\mat{H}_1^{\B})}, \trans{(\mat{H}_2^{\B})};\eta P)$  \cite{VJG03,VisTse,WSS06}.

In the remaining of this section we prove~\eqref{eq:condduality}. Notice that by Assumption~\eqref{eq:choice}, 
\begin{IEEEeqnarray}{rCl}
\mat{E} \mat{ M}_1\mat{E}&=& \mat{E}\trans{(\mat{I}+\mat{D}_1^{\B}\bBHone)}(\mat{I}+ \mat{D}_1^{\B}\bBHone) \mat{E}+ \mat{E}\trans{(\mat{D}_2^{\B}\bBHone)}(\mat{D}_2^{\B}\bBHone)\mat{E}\nonumber \\
 & = & \mat{E} \trans{(\mat{I}+\mat{E}\trans{\mat{(\mat{B}_1^{\B})}}\tBHone\mat{E})}(\mat{I}+\mat{E} \trans{\mat{(\mat{B}_1^{\B})}}\tBHone\mat{E}) \mat{E} \nonumber \\ 
 & &+ \mat{E}\trans{(\mat{E}\trans{\mat{(\mat{B}_2^{\B})}}\tBHone\mat{E})}(\mat{E}\trans{\mat{(\mat{B}_2^{\B})}}\tBHone\mat{E}) \mat{E}\nonumber \\
  & = & (\mat{I}+ \BHone \mat{B}_1^{\B})\trans{(\mat{I}+\BHone\mat{B}_1^{\B})}+( \BHone\mat{B}_2^{\B})\trans{( \BHone\mat{B}_2^{\B})}\nonumber \\
  &= & \mat{N}_1,\label{eq:MN}
\end{IEEEeqnarray}
where in the second and third equalities we used again that $\mat{E}=\trans{\mat{E}}$ and  $\mat{E}\mat{E}=\mat{I}$, and in the second equality we also used 
\begin{IEEEeqnarray}{rCl}
\bBH_1&= & \mat{I}_\eta\otimes \left( \mat{E}_\kappa \trans{\mat{H}}_{1} \mat{E}_{\nu_1}\right)\nonumber \\
&=&\left( \mat{E}_\eta \mat{I}_\eta \mat{E}_\eta\right)\otimes \left( \mat{E}_\kappa \trans{\mat{H}}_{1} \mat{E}_{\nu_1}\right)\nonumber\\
&=&(\mat{E}_\eta \otimes \mat{E}_\kappa)(\mat{I}_\eta \otimes \trans{\mat{H}}_{1})(\mat{E}_\eta \otimes \mat{E}_{\nu_1})\nonumber\\
&=& \mat{E}_{\eta \kappa} \trans{(\mat{I}_\eta \otimes \mat{H}_1)} \mat{E}_{\eta\nu_1}\nonumber\\
&=&\mat{E}_{\eta\kappa}\trans{(\BH_1)}\mat{E}_{\eta\nu_1}.
\end{IEEEeqnarray}
Here, the third and fourth equalities hold because for any matrices $\mat{A}, \mat{B}, \mat{C}, \mat{D}$ with appropriate dimensions, the Kronecker product satisfies $(\mat{A}\mat{B})\otimes(\mat{C}\mat{D}) =( \mat{A}\otimes\mat{C})(\mat{B}\otimes\mat{D})$ and $\trans{(\mat{A}\otimes\mat{B})}= \trans{\mat{A}}\otimes\trans{\mat{B}}$.

Combining~\eqref{eq:MN} with~\eqref{eq:Q1} yields 
\begin{IEEEeqnarray}{rCl}
 \mat N_1 &=&\mat{E} \mat{M}_1\mat{E}=\mat{E} \mat{Q}_1  \mat{Q}_1\mat{E}=(\mat{E} \mat{Q}_1\mat{E}) (\mat{E} \mat{Q}_1 \mat{E}).
\end{IEEEeqnarray}
Thus, by~\eqref{eq:S1} and the uniqueness  of $\mat{S}_1$, 
\begin{IEEEeqnarray}{rCl}
\mat{S}_1 & =& \mat{E} \mat{Q}_1\mat{E}.
\end{IEEEeqnarray} 
In a similar way we can also prove that
\begin{IEEEeqnarray}{rCl}
\mat{S}_2 & =& \mat{E} \mat{Q}_2\mat{E}. 
\end{IEEEeqnarray} 
Equality~\eqref{eq:condduality} follows now because for each $i\in\{1,2\}$:
\begin{IEEEeqnarray}{rCl}
\mat{E}\mat{S}_i^{-1}\BH_i\mat{E}&= & \mat{Q}_i^{-1}\mat{E}  \BH_i\mat{E}\nonumber\\
& = &\mat{Q}_i^{-1}  \trans{(\bBH_i)}\nonumber \\
& = & \trans{(\bBH_i\invtrans{\mat{Q}}_i)}\nonumber \\
& = & \trans{(\bBH_i{\mat{Q}}_i^{-1})},
 \end{IEEEeqnarray}
 where here in the last equality we used that $\mat{Q}_i$ is symmetric and thus $\mat{Q}_i^{-1}=\invtrans{\mat{Q}}_i$.

\subsection{Proof of Corollary~\ref{cor:param}}\label{sec:proofcor5}
As a first step, define the  matrices
\begin{equation}\label{eq:Ctrans}
\mat{C}_{i,\tau,\ell}' \triangleq \mat{E}\mat{C}_{i,\tau,\ell}\mat{E},
\end{equation}
and construct the strictly-lower block-triangular matrices $\mat{C}_1^{\B'}$ and $\mat{C}_2^{\B'}$ similarly to~\eqref{eq:CB}
\begin{IEEEeqnarray}{rCl}
\mat{C}_i^{\B'}&=&\begin{bmatrix}   
\mat{0}&& \dots &&\mat{0} \\ 
\mat{C}_{i,1,2}' &\mat{0}  \\
\mat{C}_{i,1,3}'&\mat{C}_{i,2,3}' &\mat{0} \\ 
\vdots&&&\ddots\\
\mat{C}_{i,1,\eta}'&\mat{C}_{i,2,\eta}'&\dots&\mat{C}_{i,(\eta-1),\eta}'&\mat{0}
\end{bmatrix},\quad i\in\{1,2\},
\end{IEEEeqnarray}
Also, let
\begin{IEEEeqnarray}{rCl}\label{eq:Dp}
 \mat{D}_i^{\B'}&\triangleq&\mat{C}_i^{\B'}\left(\mat{I}-\bBHone \mat{C}^{\B'}_1-\bBHtwo \mat{C}_2^{\B'}\right)^{-1},~i\in \{1,2\}.\IEEEeqnarraynumspace
 \end{IEEEeqnarray}

We now  show that under Assumption~\eqref{eq:equivalence_cond},
\begin{IEEEeqnarray}{rCl}
\set{R}_\BC(\eta,\mat{B}_1^{\B},\mat{B}_2^{\B},\mat{ H}_1^\B,\mat{ H}_2^\B;P)&=& 
\set{R}_\MAC(\eta,\mat{D}_1^{\B'},\mat{D}_2^{\B'},\bBHone, \bBHtwo;P) \label{eq:equal_regions12}\IEEEeqnarraynumspace
\end{IEEEeqnarray}
and moreover, 
\begin{IEEEeqnarray}{rCl}
\set{R}_\MAC(\eta,\mat{D}_1^{\B},\mat{D}_2^{\B},\tBHone, \tBHtwo;P)&=& 
\set{R}_\MAC(\eta,\mat{D}_1^{\B'},\mat{D}_2^{\B'},\bBHone, \bBHtwo;P),\label{eq:equal_regions123}\IEEEeqnarraynumspace
\end{IEEEeqnarray}
which combined establish the desired proof. 

Equation~\eqref{eq:equal_regions123} follows by Remark~\ref{rem:trans} and  because through the operation~\eqref{eq:Ctrans} the encoders transform the channel matrix $\trans{\mat{H}}_i$ into $\mat{\bar{H}}_i$. The multiplication from the left by $\mat{E}$ makes that the inputs are premultiplied by~$\mat{E}$ before they are sent over the channel and the multiplication from the right makes that the feedback outputs are first multiplied by $\mat{E}$ before further use, see~\eqref{eq:input_mac}. (See also the proof of Remark~\ref{rem:trans}.)

To prove~\eqref{eq:equal_regions12}, we shall show that 
\begin{equation}\label{eq:DpB}
\mat{\bar D}_i^{\B'}=\mat{B}_i^{\B},
\end{equation}
which by Proposition~\ref{prop:equal_regions} establishes~\eqref{eq:equal_regions12}. 
Notice first that Condition~\eqref{eq:equivalence_cond} implies 
\begin{equation}
\mat{\bar A}_i^\B= \mat{ C}_i^{\B'}.
\end{equation}
Therefore, by~\eqref{eq:Dp}, and by the properties in Note~\ref{note:bar},
\begin{IEEEeqnarray}{rCl}
\mat{\bar D}_i^{\B'}
& = & \overline{ \mat{\bar A}_i^\B  \left(\mat{I}- \bBHone \mat{\bar A}^{\B}_1 - \bBHtwo \mat{\bar A}_2^{\B}\right)^{-1}}\nonumber\\
&=&\left(\mat{I}- \mat{A}^{\B}_1\BHone - \mat{A}_2^{\B}\BHtwo\right)^{-1}\mat{A}_i^{\B}\nonumber\\
&= & \mat{B}_i^{\B} 
\end{IEEEeqnarray}
and thus concludes the proof.

\appendices
\section{Proofs of Auxiliary Results}\label{app:proofs}

\subsection{Proof of~\eqref{eq:MACsym}}\label{app:cor}
Fix a nonzero real number $h$ and a positive real number $P$.
By~\eqref{sumcapMAC}, 
\begin{IEEEeqnarray}{rCl}\label{eq:maximization}
C_{\MAC,\SISO,\Sigma}^{\fb}(h, h;P)
& = & \quad \max_{\mathclap{\substack{P_1,P_2\geq 0:\\P_1+P_2=P}}}\quad  \frac{1}{2} \log \left( 1+ h^2 P +2 h^2 \sqrt{P_1P_2}\rho^\star(h,h;P_1,P_2)\right),\nonumber \\
& = &  \max_{\alpha\in [0,1]}\frac{1}{2} \log ( 1+ h^2 P + 2 h^2 P \zeta_{P,h}(\alpha)) 
\end{IEEEeqnarray}
where the function $ \zeta_{P,h}$ is defined as
\begin{IEEEeqnarray}{rCl}\label{eq:f}
 \zeta_{P,h}\colon [0,1]&\to& \left[0,\frac{1}{4}\right] \nonumber \\
 \alpha& \mapsto &\sqrt{\alpha (1-\alpha)}\rho^\star(h,h;\alpha P,(1-\alpha) P).
 \end{IEEEeqnarray}
We argue in the following that irrespective of the values of $h$ and $P$: 
\begin{equation}\label{eq:argmax}
 \argmax_{\alpha\in[0,1]}\zeta_{P,h}(\alpha)=\frac{1}{2},
\end{equation} 
and thus the sum-capacity $C_{\MAC,\SISO,\Sigma}^{\fb}(h, h;P)$ is as in~\eqref{eq:MACsym}.
More specifically, we show that if~\eqref{eq:argmax} was violated, then 
the sum-capacity of the  scalar Gaussian MAC with symmetric channel gains $h$ and symmetric individual power constraints $P/2$ differs from $\frac{1}{2} \log ( 1+ h^2 P +  +2 h^2 P \zeta_{P,h}(1/2))$, which contradicts the results in~\cite{OZAROW84}. In fact, let's assume for contradiction that there exists a $\alpha^\star\in[0,1]$ such that 
\begin{equation}\label{eq:argmax1}
\zeta_{P,h}(\alpha^\star)>\zeta_{P,h}(1/2).
\end{equation} 
By symmetry of the function $\zeta_{P,h}$, also
\begin{IEEEeqnarray}{rCl}
\zeta_{P,h}(1-\alpha^\star) >\zeta_{P,h}(1/2).
\end{IEEEeqnarray}
We consider the following time-sharing scheme over the scalar Gaussian MAC with symmetric channel gains and power constraints. During the first half of the channel uses
we apply Ozarow's scheme~\cite{OZAROW84} where Transmitter~1 uses average power $\alpha^\star P$ and Transmitter~2 uses average power $(1-\alpha^\star)P$. During the second half we again apply Ozarow's scheme, but now Transmitter~1 uses average power $(1-\alpha^\star)P$ and Transmitter~2 uses average power $\alpha^\star P$. Over the entire block of transmission, each transmitter thus uses average power $P/2$ and satisfies the   individual average power constraint. The described scheme achieves a sum-rate of 
\begin{IEEEeqnarray}{rCl}
  R_{\Sigma}&=&\frac14 \log (1+h^2 P+2h^2 P \zeta_{P,h}(\alpha^\star)) +\frac14 \log (1+h^2 P+2h^2 P \zeta_{P,h}(1-\alpha^\star))\\
  &=&\frac12 \log (1+h^2 P+2h^2 P \zeta_{P,h}(\alpha^\star)).\label{eq:argmax2} 
\end{IEEEeqnarray}
By~\eqref{eq:argmax1} and \eqref{eq:argmax2} the rate of our scheme thus exceeds the sum-capacity of the channel under symmetric individual power constraints, which establishes the desired contradiction.

\subsection{Proof of Note~\ref{lm1}}\label{sec:pflm1}
Recall the mapping $\omega$ defined by~\eqref{eq:B}
\begin{IEEEeqnarray}{rCl}\label{eq:DC1}
 \mat{B}_i^\B&\triangleq&\left(\mat{I}- \mat{A}_1^\B\BHone- \mat{A}_2^\B\BHtwo\right)^{-1}\mat{A}_i^\B,\quad i\in \{1,2\}. 
\end{IEEEeqnarray} 
One can verify that
\begin{IEEEeqnarray}{rCl}\label{eq:DC}
 \mat{A}_i^\B&\triangleq&\left(\mat{I}+ \mat{B}_1^\B\BHone+ \mat{B}_2^\B\BHtwo\right)^{-1}\mat{B}_i^\B,\quad i\in \{1,2\}. 
\end{IEEEeqnarray} 

Observe now that: 
\begin{itemize}
\item If a matrix $\mat{A}$ is strictly-lower block-triangular  with block sizes $\kappa_1\times\kappa_2$ and a matrix $\mat{B}$ is lower block-triangular with block sizes $\kappa_2\times\kappa_3$, then the product $\mat{A}\mat{B}$ is strictly-lower block-triangular with block sizes $\kappa_1\times\kappa_3$.
\item The inverse of a lower block-triangular matrix with block sizes $\kappa$-by-$\kappa$ is again lower block-triangular with the same block sizes.
\end{itemize}
With these observations and inspecting the expressions in~\eqref{eq:DC1} and \eqref{eq:DC}, the lemma follows.

\addcontentsline{toc}{chapter}{\numberline{}Bibliography}

\end{document}